\documentclass[12pt]{amsart}
\usepackage{amsmath}
\usepackage{amsthm}
\usepackage{amssymb}
\usepackage{amsfonts}
\usepackage{paralist}
\usepackage{tikz}
\usepackage{framed}
\usepackage{graphicx}
\usepackage{hyperref}
\usepackage[margin=1in]{geometry}

\newtheorem{theorem}{Theorem}
\newtheorem{lemma}[theorem]{Lemma}
\newtheorem{corollary}[theorem]{Corollary}

\newtheorem{proposition}[theorem]{Proposition}
\newtheorem{definition}[theorem]{Definition}

\newtheorem{example}[theorem]{Example}

\newcommand{\true}{\textsc{True}}
\newcommand{\false}{\textsc{False}}

\DeclareMathOperator{\diam}{Diam}

\newcommand{\RR}{\mathbb{R}}
\newcommand{\CC}{\mathbb{C}}

\newcommand{\dist}{\mathtt{dist}}
\newcommand{\PV}{\texttt{PV}}

\newcommand{\ZZ}{\mathbb{Z}}

\newcommand{\sOB}{\widetilde{O}_B}
\newcommand{\sO}{\widetilde{O}}

\newenvironment{algorithm1}[1]
{\refstepcounter{theorem} 
\setlength{\FrameSep}{0.1in}
\MakeFramed{\setlength{\hsize}{5in} \FrameRestore}
\centerline{\bf Algorithm \arabic{theorem}. #1}\vspace{.05in}}
{\endMakeFramed}

\begin{document}

\title{The Complexity of Subdivision for Diameter-Distance Tests}

\author{Michael Burr}
\address{Department of Mathematical Sciences, Clemson University, Clemson, SC 29634}
\email{burr2@clemson.edu}
\thanks{Partially supported by a grant from the Simons Foundation (\#282399 to Michael Burr) and National Science Foundation Grant CCF-1527193.}

\author{Shuhong Gao}
\address{Department of Mathematical Sciences, Clemson University, Clemson, SC 29634}
\email{sgao@clemson.edu}
\thanks{Partially supported by the National Science Foundation  under Grants CCF-1407623, DMS-1403062 and  DMS-1547399.}

\author{Elias Tsigaridas}
\address{Sorbonne Universit{\'e}s, \textsc{UPMC} Univ Paris 06, \textsc{CNRS}, \textsc{INRIA}, Laboratoire d'Informatique de Paris 6 (\textsc{LIP6}), {\'E}quipe \textsc{PolSys}, 4 place Jussieu, 75252 Paris Cedex 05, France}
\email{elias.tsigaridas@inria.fr}
\thanks{Partially by an FP7 Marie Curie Career Integration Grant
and by the ANR JCJC ``GALOP''}

\begin{abstract}
  We present a general framework for analyzing the complexity of
  subdivision-based algorithms whose tests are based on the sizes of
  regions and their distance to certain sets (often varieties)
  intrinsic to the problem under study.  We call such tests
  diameter-distance tests.  We illustrate that diameter-distance tests
  are common in the literature by proving that many interval
  arithmetic-based tests are, in fact, diameter-distance tests.  For
  this class of algorithms, we provide both non-adaptive bounds for
  the complexity, based on separation bounds, as well as adaptive
  bounds, by applying the framework of continuous amortization.

  Using this structure, we provide the first complexity analysis for
  the algorithm by Plantinga and Vegeter for approximating real
  implicit curves and surfaces.  We present both adaptive and
  non-adaptive {\em a priori} worst-case bounds on the complexity of
  this algorithm both in terms of the number of subregions constructed
  and in terms of the bit complexity for the construction.  Finally,
  we construct families of hypersurfaces to prove that our bounds are tight.
\end{abstract}

\maketitle


\section{Introduction}
Subdivision-based algorithms are adaptive methods that start with a domain of interest (often an axis-aligned box) and recursively split it into sub-domains until each sub-domain either isolates or does not contain an interesting feature of the problem at hand.  The output is a partition of the original domain (often into axis-aligned boxes) which we can further study or post-process.  This algorithmic paradigm is one of the most commonly used classes of algorithms with appearances in many fields, ranging from computational geometry and graphics to approximating solutions to polynomial systems and mathematical programming, see, e.g.,~\cite{Lorensen:1987:MCH:37402.37422,JCC:JCC540140212,Schroder02subdivisionas,computationapplication,mmt-tcs-2010,6257651,ek-spm-01,AlErGe-ex-opt-02}.  The main goal of this paper is to study the computational complexity of these types of algorithms.

The main advantages of  subdivision-based algorithms are their great flexibility and their local nature.  Because of their recursive character, they are easy to implement using simple data structures, and this ease of use makes them popular among practitioners.  Moreover, subdivision-based algorithms are intrinsically adaptive, and they are often efficient in practice since they only perform additional subdivisions near difficult features.  These advantages, however, make the complexity analysis of subdivision-based algorithms particularly challenging.  To analyze these algorithms, we need to understand, in detail, the local complexity of the input instance and how the problem-specific predicates behave near problem-specific features because any tight complexity bound must be sensitive to the locations and sizes of easy and difficult features.

Our motivating example for this paper is the complexity analysis of the Plantinga and Vegter algorithm\footnote{Our approach applies to similar subdivision-based methods for approximating curves such as~\cite{LinYap:Cxy}.  The final complexity results are similar, and we leave the details to the interested reader.}~\cite{Plantinga:2004}.  Their algorithm is a subdivision-based algorithm for correctly approximating curves and surfaces, see Figure~\ref{fig:pv-curve-example}.  We call this algorithm the \PV\ algorithm.  It takes, as input, a polynomial $f\in\mathbb{R}[x,y]$ or $\mathbb{R}[x,y,z]$, whose real zero set is smooth\footnote{\label{fn:correctness}The correctness depends on the curve being bounded, but the termination of the algorithm depends only on the smoothness.  See \cite{Burr:Isotopic} for an extension of this algorithm which includes a correctness statements for unbounded curves.}, and an axis-aligned square $I\subseteq\mathbb{R}^2$ or cube $I\subseteq\mathbb{R}^3$.  From this input data, the algorithm constructs a piecewise-linear approximation to the zero set of $f$ in $I$.  In particular, when $I$ is a bounding box for the variety, the approximation has the correct topology in the sense that there is an ambient isotopy between the approximation and the zero set.  Additionally, by further subdivisions, the Hausdorff distance between the approximation and the zero set can be made as small as desired.  The authors of~\cite{Plantinga:2004} claim that the \PV\ algorithm is efficient in practice, but, to the best of our knowledge, the work in this paper provides the first complexity analysis of the \PV\ algorithm.  A preliminary version of this work appeared in~\cite{ComplexityRealCurves}.

\begin{figure}[hbt]
  \centering
\includegraphics{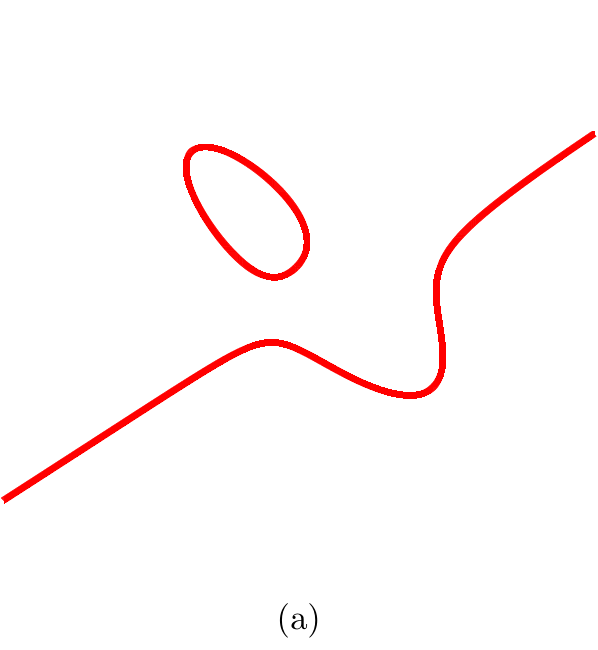}
\qquad
\includegraphics{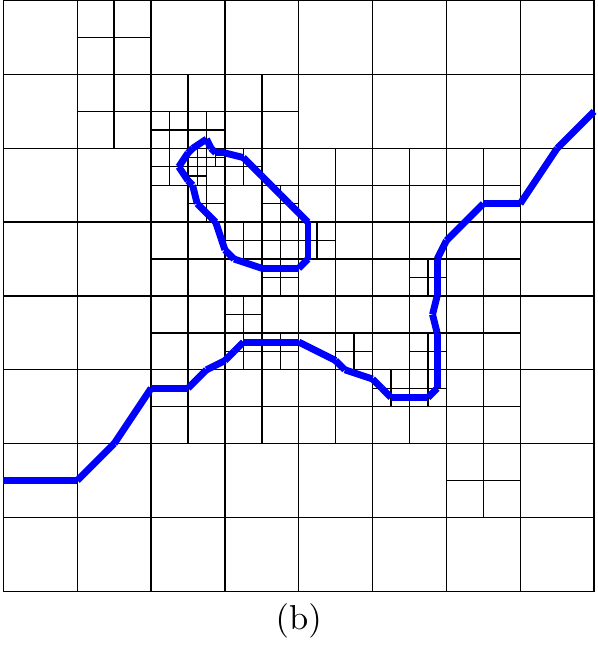}
\label{fig:pv-curve-example}
\caption{(a) The real curve given by the zeros of $f=3y^3+3xy^2-2x^3-3y^2+xy+3x^2-3y+3x+2$.
(b) The approximation produced by the \PV\ algorithm~\cite{Plantinga:2004} as well as the regions constructed by the algorithm.
}
\end{figure}

In this paper, we provide complexity bounds for the class of subdivision-based algorithms which use diameter-distance tests.  Diameter-distance tests are predicates which become more restrictive as sub-domains become closer to problem-specific subsets of the domain.  These tests are fairly common in the literature, for example, condition number-based tests are related to the inverse of the distance to the set of ill-conditioned inputs, see, e.g.,~\cite{BurCuc-condition-13}, tests for motion planning are based on the distance to the set of obstacles~\cite{SoftPredicates}, and root isolation is related to the inverse of the distance to the nearest root of the polynomial or its derivative, see, e.g.,~\cite{BurrKrahmerYap:ContinuousAmortization,Burr:2016,Eigenwillig:Descartes}.  After developing the general theory, we prove, using Fourier analysis, that the \PV\ algorithm's tests are diameter-distance tests, and, as an example, we provide the first complexity bounds for the \PV\ algorithm and prove that they are tight~\cite{Plantinga:2004}.

\subsection{Related Work}

For univariate problems, the analysis of subdivision-based algorithms
is well-understood, and there are several results, especially for the
case of approximating the roots of polynomials, see,
e.g.,~\cite{YaSaSh-arc-13,SagraloffYap:CEVAL,BurrKrahmerYap:ContinuousAmortization,Burr:2016,Eigenwillig:Descartes,SturmComplexity},
and the references therein.
Moreover, it also possible to modify the subdivision process by applying the Newton operator, see~\cite{Sag-nd-12,Pan-Weyl-00}, and considerably improve both the complexity and the actual running time of the corresponding algorithms.
However, in higher dimensions, very little is known.  For example,
there are no explicit complexity results for pure subdivision-based
algorithms for approximating curves and surfaces.

The design of efficient subdivision-based algorithms that are {\em output-sensitive}, {\em precision-sensitive}, {\em certified}, and exploit the underlying {\em structure} of the problem is an important challenge and an active area of research.  An important step in this direction was the introduction of {\em soft tests}, see~\cite{WaChYa-soft-15, YaSaSh-arc-13}, that, roughly speaking, replace harder exact tests (usually comparisons with zero) with approximate computations which are exact in the limit.  They introduce a new notion of correctness called resolution-exactness.  In this context, it is exactly the continuous amortization tool~\cite{Burr:2016,BurrKrahmerYap:ContinuousAmortization} that captures the complexity of the soft predicates. Therefore, continuous amortization is a key tool for the analysis of such algorithms.

The previous work on subdivision-based methods and inclusion-exclusion
predicates is quite extensive, so, we can only scratch its surface.
For work that focuses on classical inclusion-exclusion algorithms for
the isolation of roots of algebraic and analytic functions, but
without bit-complexity bounds, we refer the interested reader to
\cite{GiLeSaYa-focm-05,y-na-1994,dy-na-1993,YaSaSh-arc-13}, and the
references therein.  For other approaches for approximating curves and
surfaces, we refer the interested reader to
\cite{ChDeRaTa-sicomp-07,BoiCohVeg-DCG-08,BBMRV-survey-06,CheKaiLaz-jsc-13},
and the references therein.  For the problem of isolating the roots of
polynomials with subdivision-based methods, we refer the interested
reader to
\cite{mmt-tcs-2010,Krawczyk:Newton,Henrici:search:70,CollinsAkritas:Descartes,Yakoubsohn:bisection:05,MourrainPavone,BurrKrahmer:2012,SagraloffYap:CEVAL,CheGaoGuo-jsc-12},
and the references therein.  There are also approaches, see, e.g.,
\cite{MourrainPavone,mmt-tcs-2010}, that achieve locally quadratic
convergence towards the simple roots of polynomial systems, and they
are very efficient in practice.  Another interesting direction for the
application of subdivision-based algorithms, of a more geometric
nature, concerns the approximation of algebraic varieties
\cite{Snyder:Interval,Plantinga:2004,Burr:Isotopic,SharmaVegterYap:Isotopic,LinYap:Cxy,LinYapYu:Isotopic}
and the computation of the approximate Voronoi diagrams
\cite{6257651}.  There are also important applications of these
algorithms to the problem of robotic motion planning
\cite{SoftPredicates}.

For the related problem of computing the topology of an implicitly defined curve in the plane, we refer the reader to~\cite{blmprs-curv2d-16} for state-of-the-art results.  Nevertheless, we emphasize, that even though analyzing the topology of an implicitly defined curve is related to the problem we consider in this paper, the problems and approaches are different and the complexity estimates are not directly comparable.  Our approach is a general one that we can use for the analysis of any subdivision-based algorithm that uses diameter-distance tests and it is not a dedicated one for computing the topology of curves.  This argument also holds for methods and algorithms based on cylindrical algebraic decomposition, which can be used as a black-box tool to solve similar problems with curves, see~\cite{HooSaf-jsc-12} and references therein.

\subsection{Main Results}
We introduce diameter-distance tests, which formalize a type of test that is
frequently used in subdivision-based algorithms, see Section
\ref{sec:diam-dist:def}.  We then present both
straightforward non-adaptive complexity bounds for such tests based on
separation bounds, see Proposition~\ref{prop:global-size-bound}, and
adaptive bounds, based on continuous amortization, which exploit the
local features of the problem at hand in Proposition
\ref{prop:adaptive-size-bound}.  The diameter-distance tests are quite
generic in nature and we illustrate this by formulating classical
exclusion predicates, in {\em any dimension}, in conjunction with
interval arithmetic as diameter-distance tests, see Section
\ref{sec:exclusion:distance-diameter}.

We provide the first complexity analysis for (a slightly modified
version of) the \PV\ algorithm for approximating curves and surfaces
from~\cite{PlantingaVegterJournal,Plantinga:2004}.  We extend the
predicates of the \PV\ algorithm to all dimensions and bound the
number of regions and bit-complexity of these algorithms in two- and
higher-dimensions, see Theorems~\ref{thm:PV:steps} and
\ref{thm:nonadaptive-bit-complexity}.  Moreover, using continuous
amortization, first developed by
\cite{BurrKrahmerYap:ContinuousAmortization}, we provide adaptive
bounds on the number of regions and the bit-complexity of the \PV\
algorithm in arbitrary dimensions, see Theorems
\ref{theorem:PV:CA:complexity} and~\ref{thm:bit-CA}.  These results
consist of the first application of continuous amortization to a pure
high-dimensional problem.  We provide examples that show that our
bounds are tight in Lemma~\ref{lem:PV-tight}.

We anticipate that diameter-distance tests and the tools for the
complexity analysis of the underlying subdivision-based algorithms that we develop in
this paper will be applicable to many other algorithms and in related
contexts.

\subsection{Overview of Paper}

The rest of this paper is organized as follows: In the next section, we present a general description of subdivision-based algorithms, we introduce diameter-distance tests, and we derive adaptive and non-adaptive complexity bounds for subdivision-based algorithms that use these tests.  In Section~\ref{Section:localsizebound}, we show that exclusion tests based on interval arithmetic are diameter-distance tests.  This illustrates that many algorithms in the literature can be analyzed with the techniques of this paper.  In Section~\ref{sec:PV-alg}, we present a slight modification of the \PV\ algorithm for curve and surface approximation.  We then exhibit the tests in the \PV\ algorithm as diameter-distance tests.  In Section~\ref{sec:worst-case-bounds}, we present both adaptive and non-adaptive bounds on the number of subdivisions that the \PV\ algorithm performs and the bit-complexity for the overall algorithm.  Finally, in Section~\ref{sec:examples} we present examples to demonstrate the tightness of our bounds.

\section{Subdivision-based Methods and Diameter-Distance Tests}
\label{sec:subdiv-dd-test}

In Section~\ref{sec:subdiv:def}, we present the general form of a subdivision-based method which is studied in this paper.  In Section~\ref{sec:diam-dist:def}, we define the diameter-distance tests, which form the class of predicates studied in this paper.  Even though our motivating example is the Plantinga and Vegter algorithm~(\cite{Plantinga:2004,PlantingaVegterJournal}), we present this material in a general setting.  Additional, related, background on this approach for the study of subdivision-based methods in this section can be found in~\cite{BurrKrahmerYap:ContinuousAmortization} and~\cite{Burr:2016}.

Throughout this section, we assume that $X$ is both a measure space with measure $\mu$ and a metric space with distance function $d$.  We note that we do not require any compatibility between $\mu$ and $d$.  Additionally, we assume the technical condition that $X$ is proper\footnote{The theory continues to apply even with weaker conditions, but there are a few technicalities that arise.  For most applications, this assumption is not an additional constraint.  For weaker conditions, we leave the details to the interested reader.}, i.e., closed balls are compact.  Moreover, we let $\mathcal{S}$ be a collection of subsets of $X$ which have finite measure (with respect to $\mu$) and are compact (with respect to $d$).  We call predicates (boolean functions) on $\mathcal{S}$ {\em stopping criteria}.  In the motivating case, $X=\mathbb{R}^n$, $\mu$ is the Lebesgue measure, $d$ is the Euclidean distance, and $\mathcal{S}$ is the collection of $n$-dimensional cubes in $X$.

\subsection{Subdivision-based Methods}\label{sec:subdiv:def}

In this section, we provide the general form of a subdivision-based method considered in this paper.  Let $C_1,\dots,C_\ell$ be stopping criteria, i.e., each $C_i$ is a function from subsets in $\mathcal{S}$ to $\{\true,\false\}$.  When $C_i(J)$ is $\true$ for some $i$, then we do not split $J$, and, when $C_i(J)$ is $\false$ for all $i$, then we must subdivide $J$.

We use the following simple abstract algorithm to describe subdivision-based tests.  Fix $C_1,\dots,C_\ell$ to be stopping criteria, and consider an input region $I\in \mathcal{S}$.  The output of the algorithm is a partition $P$ of $I$ such that for each element $J$ in $P$, there is some $i$ so that $C_i(J)=\true$.  Initially, $P=\{I\}$.

\begin{algorithm1}{Abstract Subdivision-based Algorithm}\label{algorithm:abstract}
\noindent Repeatedly subdivide each $J\in P$ until:\\[.15cm]
\indent There exists $1\leq i\leq \ell$ so that $C_i(J)=\true$.
\end{algorithm1}

To subdivide a region $J$ means to replace $J$ in $P$ with regions
$J_1,\dots,J_k$, where $k\geq 2$, each $J_j\in\mathcal{S}$,
$J=\cup_j J_j$, and the pairwise intersections of the $J_j$ are
measure zero subsets.  In this paper, we add two mild
additional assumptions to the subdivisions under consideration: Let
$0<\varepsilon_1,\varepsilon_2<1$.  Then, we add the following
assumptions:
\begin{enumerate}
\item[Assumption 1:] $\mu(J_j)\geq\varepsilon_1\mu(J)$ and
\item[Assumption 2:] $\diam(J_j)\leq\varepsilon_2\diam(J)$.
\end{enumerate}
The first condition prevents $J$ from splitting into two many regions at any step, while the second condition generalizes the idea that the aspect ratio of $J_i$ should not be too large.  For additional details on the first assumption, see~\cite[Lemma 3.5 and Remark 3.6]{Burr:2016}.  The {\em subdivision tree} is the tree whose root is $I$, whose internal nodes represent sub-domains $J$ that are processed during the subdivision, whose leaves are the terminal regions, and where the parent-child relationship is given by subdivision.  We observe that, in the motivating example for this paper, $\varepsilon_1=2^{-n}$ and $\varepsilon_2=2^{-1}$.

\subsection{Diameter-distance Tests}\label{sec:diam-dist:def}

In this section, we define distance-diameter tests.  These tests form the class of predicates that we consider in this paper.  Many tests which have been developed, such as the one-circle condition in Descartes' rule of signs~\cite{OneCircleTest,AleGal-Vincent-00}, are diameter-distance tests.  At first glance, it might appear that the definition of diameter-distance tests is very specialized; this is not the case.
In fact, in Section~\ref{Section:localsizebound}, we provide a nontrivial example of these tests which appear frequently in applications.

\begin{definition}
Let $X$ and $\mathcal{S}$ be defined as above.  Let $C$ be a stopping criterion on $\mathcal{S}$.  $C$ is a {\em diameter-distance test} if there exists a closed set $V\subseteq X$ and a positive constant $K$ such that for any $J\in\mathcal{S}$,
$$
\text{If }\left(\diam(J)<K\max_{x\in J}d(x,V)\right)\text{, then }C(J)=\true,
$$
where $d(x,V)=\min_{v\in V}d(x,v)$.
\end{definition}

The extra conditions, such as compactness and properness allow one to use the minimum and maximum in the definition above instead of infimum and supremum.  Loosely speaking, this definition states that $C(J)$ must be $\true$ when $J$ includes a point sufficiently far away from $V$ and $J$ isn't too large.

We note that this definition does not state that the stopping criterion $C$ must be a distance-based test or even that $V$ is known to $C$.  Instead, the only assumption is that the criterion is less conservative than the conditional in the definition above.  In particular, stopping criteria whose theory is based upon condition numbers are frequently diameter-distance tests because the condition number can be rewritten in terms of the inverse of the distance to the set of ill-conditioned inputs~\cite{BurCuc-condition-13}.

Throughout the remainder of this section, we assume that all stopping criteria are diameter-distance tests.

\subsection{Non-adaptive Bounds}

In this section, we provide a lower bound on the number of regions produced by Algorithm~\ref{algorithm:abstract}.  This analysis is not adaptive, so it assumes the worst-case behavior everywhere.  We include this approach for comparison because the adaptive bounds are based on the ideas of the non-adaptive bounds and, in some cases, the adaptive bounds may be too complicated to compute.  In the next section, we provide an adaptive bound based on continuous amortization~\cite{Burr:2016}.

\begin{proposition}\label{prop:global-size-bound}
Suppose that the stopping criteria $C_1,\dots,C_\ell$ in Algorithm~\ref{algorithm:abstract} are all diameter-distance tests, with associated positive constants $K_i$ and closed subsets $V_i$.  Furthermore, assume that the intersection $\bigcap V_i$ is empty.  Let $K=\min K_i$, and let $I$ be the initial input region.  Define the separation bound $\delta$ as 
$$
0<\delta\leq\min_{x\in I}\max_i d(x,V_i).
$$
Then, the number of regions constructed by the algorithm is at most
$$
\max\left\{1,\varepsilon_1^{-1+\frac{\ln\diam(I)-\ln(K\delta)}{\ln(\varepsilon_2)}}\right\}.
$$
\end{proposition}
Before presenting  the proof, we note that $\delta$ is a lower bound on the smallest distance from any $x\in I$ to the furthest $V_i$.  We call $\delta$ a separation bound because if the $V_i$'s are pairwise disjoint and $\Delta$ is the minimum distance between them, then $\frac{\Delta}{2}$ satisfies the conditions for $\delta$.
\begin{proof}
If no subdivisions occur, then the only region is the initial region.  In this case, $1$ region is constructed and the bound holds.  We, therefore, assume that subdivisions occur.  Let $J_j$ be a terminal region and $J$ its parent; moreover, let $x\in J_j$.  Since $x\in J$ and $J$ was subdivided, we know that $\diam(J)\geq K\max_i d(x,V_i)$, since, otherwise, by the definition of a diameter-distance test, for some $i$, $C_i(J)$ would be \true, contradicting the assumption that $J$ was subdivided.  Since this maximum is larger than $\delta$, it follows that $\diam(J)\geq K\delta$.

Suppose that the depth of $J_j$ in the tree is $k$, then the depth of $J$ is $k-1$, and, by the assumption on diameters under subdivisions, we know that $\diam(J)\leq\varepsilon_2^{k-1}\diam(I)$.  Therefore, $K\delta\leq\varepsilon_2^{k-1}\diam(I)$.  Taking the logarithm of both sides (and recalling that $\ln(\varepsilon_2)<0$), it follows that
$$
k\leq1+\frac{\ln(K\delta)-\ln(\diam(I))}{\ln(\varepsilon_2)}.
$$

By the assumption on volumes for subdivisions, it follows that $\mu(J_j)\geq\varepsilon_1^k\mu(I)$.  Substituting in our expression for $k$, we can conclude that 
$$
\mu(J_j)\geq \varepsilon_1^{1+\frac{\ln(K\delta)-\ln(\diam(I))}{\ln(\varepsilon_2)}}\mu(I).
$$
This lower bound applies to the measure of every terminal region.  Moreover, since the pairwise intersection of terminal regions has zero measure, we know that, in the worse-case, $I$ is subdivided into regions of this size, which results in the desired bound.
\end{proof}

We observe that in the motivating case, i.e., where $\varepsilon_1=2^{-n}$ and $\varepsilon_2=2^{-1}$, the maximum above simplifies to
\begin{equation}\label{eq:nonadaptive:bound}
\max\left\{1,\left(\frac{2\sqrt{n}}{K\delta}\right)^n\mu(I)\right\}
\end{equation}
since $\diam(I)^n=n^{n/2}\mu(I)$.  In many cases, however, this bound is much larger than necessary as the analysis assumes that the worst-case situation occurs everywhere.  An adaptive bound is necessary to account for this non-uniformity.

\subsection{Adaptive Bounds}

In this section, we present adaptive bounds for the number of regions produced by Algorithm~\ref{algorithm:abstract}.  This adaptive bound is based on the continuous amortization technique~\cite{Burr:2016}, which we briefly review here.  

Continuous amortization was introduced in~\cite{BurrKrahmerYap:ContinuousAmortization} as a way to adaptively analyze the complexity of subdivision-based algorithms.  In~\cite{Burr:2016}, the theory of continuous amortization was extended to apply to measure spaces and to evaluate functions on the regions of the partition, and we recall this technique here.  The key to continuous amortization is a function on $X$, called a local size bound, which is a point estimate, locally describing the worst-case amount of work that is required at each point.

\begin{definition}\label{def:localsizebound}
Let $X$ and $\mathcal{S}$ be defined as above and $C$ a stopping criterion on $\mathcal{S}$.  A {\em local size bound} for $C$ is a function $F:X\rightarrow\mathbb{R}_{\geq 0}$ with the property that 
$$
F(x)\leq \inf_{\substack{J\in\mathcal{S}\\J\ni x\\C(J)=\textsc{False}}}\mu(J).
$$
\end{definition}

In other words, $F(x)$ is a lower bound on the measure of a region which contains $x$, but fails the stopping criterion.  The local size bound provides the link between the algorithm and a quantity that we can compute.

\begin{theorem}[{\cite{BurrKrahmerYap:ContinuousAmortization,Burr:2016}}]\label{thm:local-sz-bd}    
Let $X$ and $\mathcal{S}$ be defined as above, $C$ a stopping criterion on $\mathcal{S}$, and $F$ a local size bound for $C$.  Let $h:\mathbb{R}_{\geq 0}\rightarrow\mathbb{R}$ be a non-increasing function, and let $P$ be the final partition formed by Algorithm~\ref{algorithm:abstract}, which recursively subdivides the input region $I$, subject to Assumption 1.  The sum of $h$ applied to the regions in $P$ is bounded as follows:
$$
\sum_{J\in P}h(\mu(J))\leq\max\left\{h(\mu(I)),\int_I\frac{h(\varepsilon_1F(x))}{\varepsilon_1F(x)}d\mu\right\}.
$$
If $h\equiv 1$, i.e., $h$ is the constant function, then this integral counts the number of regions formed by Algorithm~\ref{algorithm:abstract}.  In addition, if the algorithm does not terminate, then the integral is infinite.
\end{theorem}

We observe that we can use the continuous amortization integral to express the complexity of Algorithm~\ref{algorithm:abstract} in the particular case where each stopping criterion is a diameter-distance test.

\begin{lemma}\label{lem:diameter-distance:lsb}
Let $X$ and $\mathcal{S}$ be as above and $C$ a stopping criterion on $\mathcal{S}$ with constant $K$ and closed set $V$.  Suppose that the subdivisions by Algorithm~\ref{algorithm:abstract} are subject to the two additional conditions following Algorithm~\ref{algorithm:abstract}.  Then, 
$$
F(x)=\varepsilon_1^{1+\frac{\ln(Kd(x,V))-\ln(\diam(I))}{\ln(\varepsilon_2)}}\mu(I)
$$
is a local size bound for $C$.
\end{lemma}
\begin{proof}[Proof Sketch]
The proof is identical to that of Proposition~\ref{prop:global-size-bound} except that we begin with the condition that $\diam(J)\geq Kd(x,V)$ from the definition of a diameter-distance test.
\end{proof}

In Algorithm~\ref{algorithm:abstract}, we have multiple stopping criteria.  Therefore, for each $C_i$, we can define a local size bound $F_i:X\rightarrow\mathbb{R}_{\geq 0}$.  Moreover, since at least one of the stopping criteria must be true, we can take the maximum of all of them for the local size bound for Algorithm~\ref{algorithm:abstract}.  In particular, we have the following result:

\begin{proposition}\label{prop:adaptive-size-bound}
Suppose that the stopping criteria $C_1,\dots,C_\ell$ in Algorithm~\ref{algorithm:abstract} are all diameter-distance tests with associated positive constants $K_i$ and closed subsets $V_i$.  Furthermore, assume that the intersection $\bigcap V_i$ is empty.  Let $I$ be the initial input region.  Then, the number of regions constructed by the algorithm is at most 
$$
\max\left\{1,\mu(I)^{-1}\int_I\min_i\left\{\varepsilon_1^{-1+\frac{\ln(\diam(I))-\ln(K_id(x,V_i))}{\ln(\varepsilon_2)}}\right\}d\mu\right\}.
$$
\end{proposition}

We observe that in the motivating case, i.e., where $\varepsilon_1=2^{-n}$ and $\varepsilon_2=2^{-1}$, the continuous amortization integral simplifies to
$$
\max\left\{1,\int_I\min_i\left(\frac{2\sqrt{n}}{K_id(x,V_i)}\right)^nd\mu\right\}.
$$
We apply both the adaptive and non-adaptive bounds to the \PV\ algorithm as a specific example in Section~\ref{sec:worst-case-bounds}.

\section{Interval Methods and Diameter-Distance Tests}\label{Section:localsizebound}

In this section, we show that a common exclusion test which is based on the standard centered form is a diameter-distance test.  We begin with a brief review of the standard centered form, for more details, see, for example,~\cite{Ratschek:1984,MooreInterval}.

Let $Y$ be any set, $\mathcal{S}$ a collection of subsets of $Y$, and consider the function $f:Y\rightarrow\mathbb{R}$.  An {\em interval method} for $f$ is an algorithm $\square f$ such that for any subset $J\in\mathcal{S}$, $\square f(J)\supseteq f(J)$, where $f(J)$ is the image of $J$ under $f$.  In other words, $\square f(J)$ is an over-approximation for the image $f$ on $J$.  In most applications, $Y$ is a metric space, and we add the convergence condition for $\square f$, i.e., that for a sequence of domains $\{J_k\}$ which converge to a point $p$, then $\{\square f(J_k)\}$ converges to $f(p)$.

In our applications, we consider the case where $Y$ is $\mathbb{R}^n$, and the regions in $\mathcal{S}$ are axis aligned $n$-dimensional boxes, i.e., for $J\in\mathcal{S}$, $J=\prod_i [a_i,b_i]$.  In this case, most interval methods use {\em interval arithmetic}, i.e., arithmetic operations on intervals that produce the set-theoretic image as an interval.  In this section, we focus on the standard centered form for multivariate polynomials~\cite{Ratschek:1984,MooreInterval}.  Let $f\in\mathbb{R}[x_1,\dots,x_n]$ be a multivariate polynomial of total degree $d$ and $J$ an axis aligned box.  Let $m=m(J)$ be the midpoint of $J$, then the standard centered form for $f$ applied to $J$ is
$$
\square f(J)=f(m)+\sum_{|\alpha|=1}^d \frac{\partial^\alpha f(m)}{\alpha!}(J-m)^\alpha,
$$
where $\alpha\in\mathbb{N}^n$ and the notation is multi-index notation, i.e., $|\alpha|=\sum \alpha_i$, $\partial^\alpha f(m)=\partial_1^{\alpha_1}\dots\partial_n^{\alpha_n}f(m)$, $\alpha!=\prod (\alpha_i)!$, and $(J-m)^\alpha=\prod \left[\frac{a_i-b_i}{2},\frac{b_i-a_i}{2}\right]^{\alpha_i}$.  Since $(J-m)^{\alpha}$ is a product of intervals centered at zero, using interval arithmetic, this product simplifies to $\prod\left(\frac{b_i-a_i}{2}\right)^{\alpha_i}[-1,1]$.  In the special case where $J$ is an axis-aligned, $n$-dimensional cube, i.e., $J$ is a product of $n$ intervals all of the same width $w$, then all of the factors in $(J-m)^\alpha$ are identical, and the standard centered form can be rewritten as:
\begin{equation}\label{eq:simplified:centered:form}
\square f(J)=f(m)+\left(\sum_{|\alpha|=1}^d \frac{\left|\partial^\alpha f(m)\right|}{\alpha!}\left(\frac{w}{2}\right)^{|\alpha|}\right)[-1,1],
\end{equation}
The standard centered form is an interval version of a multivariate Taylor expansion centered at $m$, and the standard centered form has several nice properties including a very structured expression and fast convergence.

In the remainder of this section, we consider the following predicate:
$$
C(J)=\true\quad\text{if and only if}\quad 0\not\in\square f(J).
$$
If $0\not\in\square f(J)$, then we can directly conclude that $0\not\in f(J)$; we observe that the converse does not hold in general, but converges in the limit, i.e., if $\{J_k\}$ is a sequence of $n$-dimensional boxes whose limit is $p$, then either $C(J_k)=\false$ for some $k$ or $f(p)=0$.  It is often more efficient, in practice, to test $C(J)$ and subdivide $J$, if necessary, rather than to compute $f(J)$ directly.  In the remainder of this section, we prove that $C(J)$ is a diameter-distance test.  

\subsection{Bounds on Coefficients of Powers of Sines and Cosines}
\label{sec:sin-cos}

In this section, we prove a technical lemma on the magnitudes of the coefficients of sines and cosines.  The main result in this section is used in the following section to prove that $C(I)$ is a diameter-distance test.

\begin{lemma}\label{lemma:sinebound}
Suppose that for all $\theta$,
\begin{equation}\label{equation:sinebound}
\left|\sum_{j=0}^ka_j\cos^j(\theta)\sin^{k-j}(\theta)\right|\leq C.
\end{equation}
Then, $|a_j|\leq 2^{k+1}C$.
\end{lemma}
\begin{proof}
Let $f(\theta)=\sum_{j=0}^ka_j\cos^j(\theta)\sin^{k-j}(\theta)$.  We observe that this is a square-integrable and $2\pi$-periodic function.  Moreover, its Fourier coefficients are bounded by $2 C$ since $|f(\theta)\cos(n\theta)|\leq C$ and $|f(\theta)\sin(n\theta)|\leq C$.

We now observe that $\cos(x)=\frac{1}{2}(e^{ix}+e^{-ix})$ and $\sin(x)=\frac{1}{2i}(e^{ix}-e^{-ix})$.  Therefore, 
$$
\cos^j(\theta)\sin^{k-j}(\theta)=\frac{1}{2^ki^{k-j}}(e^{ix}+e^{-ix})^j(e^{ix}-e^{-ix})^{k-j}=\frac{1}{2^ki^{k-j}}\sum_{l=0}^kb_le^{i(k-2l)x},
$$
where the $b_l$'s are sums and products of binomial coefficients.  Since $e^{i(k-2l)x}=\cos((k-2l)x)+i\sin((k-2l)x)$, it follows that $\cos^j(\theta)\sin^{k-j}(\theta)$ has a finite Fourier series whose nonzero terms are of the form $\cos(n\theta)$ and $\sin(n\theta)$ where $0\leq n\leq k$ and $k-n$ is even.  Therefore, the Fourier series of $f(\theta)$ is can be written as follows:
\begin{equation}\label{eq:mainFourier}
f(\theta)=\frac{c_0}{2}+\sum_{l=0}^{\lfloor (k-1)/2\rfloor}(c_{k-2l}\cos((k-2l)\theta)+d_{k-2l}\sin((k-2l)\theta))
\end{equation}
where all the constants are bounded by $2C$.

Suppose that $k\geq n$ and $k-n$ is even.  Then, since $\cos(nx)=\Re((\cos(x)+i\sin(x))^n)$ and $\sin(nx)=\Im((\cos(x)+i\sin(x))^n)$, we have the following\footnote{We note that the expression for $\cos(nx)$ is a multiple of of the Chebyshev polynomials of the first kind.}:
\begin{align*}
\cos(nx)&=(\sin^2(x)+\cos^2(x))^{\frac{k-n}{2}}\left(\sum_{m=0}^{\lfloor n/2\rfloor}(-1)^m\binom{n}{2m}\cos^{n-2m}(x)\sin^{2m}(x)\right)\\
  \sin(nx)&=(\sin^2(x)+\cos^2(x))^{\frac{k-n}{2}}\left(\sum_{m=0}^{\lfloor (n-1)/2\rfloor}(-1)^m\binom{n}{2m+1}\cos^{n-1-2m}(x)\sin^{2m+1}(x)\right).
\end{align*}
We observe that in these expansions, $\cos(nx)$ and $\sin(nx)$ are written as a linear combination of products of sines and cosines of degree $k$.  Reorganizing these sums, we find that
\begin{equation}\label{eq:cos(nx)}
  \cos(nx)=\sum_{m=0}^{\left\lfloor \frac{k}{2}\right\rfloor}\left[\sum_{p=\max\left\{0,m-\frac{k-n}{2}\right\}}^{\min\left\{\left\lfloor \frac{n}{2}\right\rfloor,m\right\}}(-1)^p\binom{n}{2p}\binom{\frac{k-n}{2}}{m-p}\right]\cos^{k-2m}(x)\sin^{2m}(x)
\end{equation}
  and
\begin{equation}\label{eq:sin(nx)}
 \sin(nx)=\sum_{m=0}^{\left\lfloor \frac{k-1}{2}\right\rfloor}\left[\sum_{p=\max\left\{0,m-\frac{k-n}{2}\right\}}^{\min\left\{\left\lfloor \frac{n-1}{2}\right\rfloor,m\right\}}(-1)^p\binom{n}{2p+1}\binom{\frac{k-n}{2}}{m-p}\right]\cos^{k-2m-1}(x)\sin^{2m+1}(x).
\end{equation}

We observe that in the formula for $\cos(nx)$, the coefficient of $\cos^{k-2l}(x)\sin^{2l}(x)$ can be bounded as follows:
\begin{align}
  \left|\sum_{p=\max\{0,m-(k-n)/2\}}^{\min\{\lfloor n/2\rfloor,m\}}(-1)^p\binom{n}{2p}\binom{\frac{k-n}{2}}{m-p}\right|&\leq\sum_{p=0}^{\lfloor n/2\rfloor}\binom{n}{2p}\sum_{q=0}^{(k-n)/2}\binom{\frac{k-n}{2}}{q}\notag\\
&\hspace*{-4cm}=\sum_{p=0}^{\lfloor n/2\rfloor}\left(\binom{n-1}{2p}+\binom{n-1}{2p-1}\right)\sum_{q=0}^{(k-n)/2}\binom{\frac{k-n}{2}}{q}\leq 2^{n-1}2^{(k-n)/2}=2^{\frac{n+k}{2}-1}\label{eq:upper:cos}
\end{align}
Similarly, the coefficient of $\cos^{k-2l-1}(x)\sin^{2l+1}(x)$ in the formula for $\sin(nx)$ is bounded by
\begin{align}
  \left|\sum_{p=\max\{0,m-(k-n)/2\}}^{\min\{\lfloor (n-1)/2\rfloor,m\}}(-1)^p\binom{n}{2p+1}\binom{\frac{k-n}{2}}{m-p}\right|&\leq\sum_{p=0}^{\lfloor (n-1)/2\rfloor}\binom{n}{2p+1}\sum_{q=0}^{(k-n)/2}\binom{\frac{k-n}{2}}{q}\notag\\
  &\hspace*{-5.4cm}=\sum_{p=0}^{\lfloor (n-1)/2\rfloor}\left(\binom{n-1}{2p+1}+\binom{n-1}{2p}\right)\sum_{q=0}^{(k-n)/2}\binom{\frac{k-n}{2}}{q}\leq 2^{n-1}2^{(k-n)/2}=2^{\frac{n+k}{2}-1}.\label{eq:upper:sin}
\end{align} 
Moreover, we observe that these bounds are independent of $m$ and $p$, depending only on $n$ and $k$

In order to bound the coefficients $a_j$, we substitute the formulas above for $\cos(nx)$ and $\sin(nx)$ into the Fourier series for $f$.  In particular, we substitute $n=k-2l$ into $2^{\frac{n+k}{2}-1}$ to get $2^{k-l-1}$.  Moreover, by considering the powers in Equations~(\ref{eq:cos(nx)}) and~(\ref{eq:sin(nx)}), we conclude that if $k-j$ is even, then $\cos^j(\theta)\sin^{k-j}(\theta)$ only appears in the expansion of the cosine terms (perhaps including the constant term) in the Fourier series for $f(\theta)$, while if $k-j$ is odd, then $\cos^j(\theta)\sin^{k-j}(\theta)$ only appears in the expansion of the sine terms in the Fourier series for $f(\theta)$.  We can then isolate the occurrences of $\cos^j(\theta)\sin^{k-j}(\theta)$ (there are four cases, depending on the parity of $k$ and $j$).  Then, using the triangle inequality and the upper bounds in Inequalities~(\ref{eq:upper:cos}) and~(\ref{eq:upper:sin}), we find that 
$$
|a_j|\leq 2C\sum_{l=0}^{\left\lfloor\frac{k}{2}\right\rfloor}2^{k-l-1}<2^{k+1}C,
$$
which completes the proof.
\end{proof}

\begin{corollary}\label{corollary:multifourier}
Fix $k_0\in\mathbb{N}$, and suppose that for all $\theta_1$, $\dots$, $\theta_m$, 
$$
\left|\sum_{k_1=0}^{k_0}\sum_{k_2=0}^{k_1}\dots\sum_{k_m=0}^{k_{m-1}}a_{(k_1,\dots,k_m)}\prod_{j=1}^m\left(\sin^{k_{j-1}-k_j}(\theta_j)\cos^{k_j}(\theta_j)\right)
\right|\leq C.
$$
Then, $a_{(k_1,\dots,k_m)}\leq 2^{m(k_0+1)} C$.
\end{corollary}
\begin{proof}
Proof by induction on $m$; the base case is Lemma~\ref{lemma:sinebound}.  For the inductive case, we fix $\theta_2,\dots,\theta_m$.  For each $k_1$, we define 
$$
a_{k_0-k_1}= \sum_{k_2=0}^{k_1}\dots\sum_{k_m=0}^{k_{m-1}}a_{(k_1,\dots,k_m)}\prod_{j=2}^m\left(\sin^{k_{j-1}-k_j}(\theta_j)\cos^{k_j}(\theta_j)\right).
$$
Then, the given inequality simplifies to
$$
\left|\sum_{k_1=0}^{k_0}a_{k_0-k_1}\sin^{k_0-k_1}(\theta_1)\cos^{k_1}(\theta_1)\right|<C.
$$
By Lemma~\ref{lemma:sinebound}, $|a_{k_0-k_1}|\leq 2^{k_0+1} C$.  Since $\theta_2,\dots,\theta_m$ are fixed, but arbitrary, and the bound does not depend on $\theta_2,\dots,\theta_m$, we can apply the inductive hypothesis to $a_{k_0-k_1}$ to give that $|a_{(k_1,\dots,k_m)}|\leq 2^{(m-1)(k_1+1)}|a_{k_0-k_1}|\leq 2^{(m-1)(k_1+1)+(k_0+1)} C$.  Since $k_1\leq k_0$, the claim follows.
\end{proof}

\subsection{Exclusion Interval Arithmetic Tests are Distance-Diameter Tests}\label{sec:exclusion:distance-diameter}

In this section, we use the results of Section~\ref{sec:sin-cos} to prove that the predicate on $n$-dimensional cubes\footnote{This analysis can be extended to the case of a region which is not an $n$-dimensional cube, by considering the smallest $n$-dimensional cube containing the $n$-dimensional box.  We leave the details to the interested reader.} $J$ where $C(J)=\true$ if and only if $0\not\in\square f(J)$ is a distance-diameter test.  In this case, the set in the definition of a distance-diameter test is the complex variety $V_\mathbb{C}(f)$.  As our first step, we reduce a higher-dimensional problem to a collection of one-dimensional problems as follows:

\begin{definition}
Let $f\in\mathbb{R}[x_1,\dots,x_n]$, $p\in\mathbb{R}^n$, and $v\in S^{n-1}$.  We define $f_v(t)$ to be the univariate polynomial passing through $p$ and in the direction $v$, i.e., $f_v(t)=f(p+tv)$, see Figure~\ref{fig:fp}.  Next, we define $\Sigma_{f_v}$ to be the sum of the reciprocals of the complex roots of $f_v$, i.e.,
$$
\Sigma_{f_v}(p)=\sum_{s\in V_\mathbb{C}(f_v)}\frac{1}{|s|}.
$$
\end{definition}

\begin{figure}[hbt]
  \centering
\includegraphics{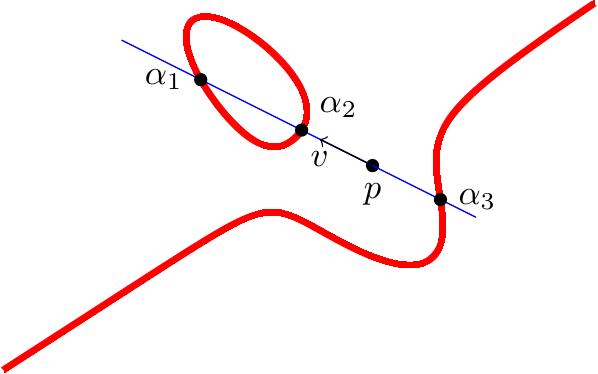}
\caption{For a polynomial $f \in \RR[x, y]$ and a point $p \in \RR^2$, we consider the roots of $f$, $\alpha_1, \alpha_2, \alpha_3$, in the direction of a unit vector $v$.}
\label{fig:fp}
\end{figure}

In~\cite[Lemma 2.1]{BurrKrahmer:2012}, it was shown that for a univariate polynomial $g\in\mathbb{R}[x]$ with complex roots $V_\mathbb{C}(g)$,
$$
\left|\frac{g^{(n)}(x)}{g(x)}\right|\leq\left(\sum_{\alpha\in V_\mathbb{C}(g)}\frac{1}{|x-\alpha|}\right)^n.
$$
This link between the Taylor coefficients of $g$ to the geometry of the zero set of $g$ can be extended to the current setting since $f_v$ is a univariate polynomial.  We introduce the notation $\dist_{\CC}(p,f)$ to represent the {\em complex} distance between the point $p$ and the variety $V_{\CC}(f)$.  Explicitly, we have the following lemma:

\begin{lemma}\label{lemma:distancebound}
Let $f\in\mathbb{R}[x_1,\dots,x_n]$, $p\in\mathbb{R}^n$, and $v\in S^{n-1}$.  Then
$$
\left|\left.\frac{1}{f(p)}\frac{d^kf(p+tv)}{dt^k}\right|_{t=0}\right|\leq\left(\Sigma_{f_v}(p)\right)^k\leq\left(\frac{\deg(f)}{\dist_{\CC}(p,f)}\right)^k.
$$
\end{lemma}
\begin{proof}
The claim is trivial when $k=0$.  Since $f_v$ is a univariate polynomial, the first inequality follows directly from~\cite[Lemma 2.1]{BurrKrahmer:2012}.  The second inequality follows from the fact that in the sum for $\Sigma_{f_v}(p)$, there are at most $\deg(f)$ terms, and each element of the sum is the inverse of the distance between $p$ and a point on $V_{\CC}(f)$, each of which is, in turn, bounded above by the inverse of the distance to the closest point on $V_{\CC}(f)$.
\end{proof}

We now use this upper bound along with the results from Section~\ref{sec:sin-cos} to bound individual Taylor coefficients in the multivariate Taylor expansion.

\begin{proposition}\label{A:proposition:derivativebound}
Let $f\in\mathbb{R}[x_1,\dots,x_n]$ and $p\in\mathbb{R}^n$.  Then, for all multi-indices $\alpha\in\mathbb{N}^n$ with $k=|\alpha|$, 
$$
\left|\frac{1}{f(p)}\binom{k}{\alpha}\frac{\partial^{k}f}{\partial x^\alpha}(p)\right|\leq2^{(n-1)(|\alpha|+1)}\left(\frac{\deg(f)}{\dist_{\CC}(p,f)}\right)^{|\alpha|}.
$$
where $\binom{k}{\alpha}$ is the multinomial coefficient.
\end{proposition}

\begin{proof}
Let $(\theta_1,\dots,\theta_{n-1})\in\left(S^1\right)^{n-1}$.  Consider the surjective map $\left(S^1\right)^{n-1}\rightarrow S^{n-1}$ given by
$
(\theta_1,\dots,\theta_{n-1})\mapsto x=(x_0,\dots,x_{n-1})
$
where
$$
x_i=\left(\prod_{j=1}^i\cos\theta_j\right)\sin\theta_{i+1}\quad\text{for}\quad0\leq i<n-1\quad\text{and}\quad x_{n-1}=\prod_{j=1}^{n-1}\cos\theta_j.$$
Then, for $k_0=k\geq k_1\geq k_2\geq\dots\geq k_{n-1}\geq 0=k_n$, 
\begin{equation}\label{equation:powersofx}
x_0^{k_0-k_1}x_1^{k_1-k_2}\dots x_{n-1}^{k_{n-1}-k_n}=\prod_{j=1}^{n-1}(\sin^{k_{j-1}-k_j}\theta_j\cos^{k_j}\theta_j).
\end{equation}

Observe that, by the chain rule for $v\in S^{n-1}$,
$$
\frac{1}{f(p)}\frac{d^kf(p+tv)}{dt^k}=\frac{1}{f(p)}\sum_{|\alpha|=k}\binom{k}{\alpha}\frac{\partial^k f}{\partial x^\alpha}(p)v^\alpha.
$$
By Lemma~\ref{lemma:distancebound}, we know that the magnitude of these expressions are bounded above by $\left(\frac{\deg(f)}{\dist_{\CC}(p,f)}\right)^k$.  Moreover, since $v$ is a unit vector, there exist $k_0=k\geq k_1\geq k_2\geq\dots\geq k_{n-1}\geq 0=k_n$ so that $v^\alpha$ can be written in the form of Equation~(\ref{equation:powersofx}), 
$$
v^\alpha=\prod_{j=1}^{n-1}(\sin^{k_{j-1}-k_j}\theta_j\cos^{k_j}\theta_j).
$$
Therefore, 
$$
\left|\frac{1}{f(p)}\sum_{|\alpha|=k}\binom{k}{\alpha}\frac{\partial^k f}{\partial x^\alpha}(p)v^\alpha\right|
$$
is of the form for Corollary~\ref{corollary:multifourier} where $m=n-1$ and $k_0=k$.  Therefore, the individual terms are bounded by $2^{(n-1)(k+1)}\left(\frac{\deg(f)}{\dist_{\CC}(p,f)}\right)^k$, as desired.
\end{proof}

With Proposition~\ref{A:proposition:derivativebound} in hand, we now prove a lower bound on an $n$-dimensional cube $J$ of width $w=w(J)$ that fails the predicate $C(J)$ in the following corollary:

\begin{corollary}\label{A:corollary:testbound}
Let $f\in\mathbb{R}[x_1,\dots,x_n]$ and $p\in\mathbb{R}^n$.  Suppose that $0<w\leq\frac{\dist_{\CC}(p,f)\ln(1+2^{2-2n})}{2^{n-1}\deg(f)}$.  Then
$$
\left|\sum_{k=1}^{\deg(f)}\sum_{|\alpha|=k}\frac{1}{k!}\binom{k}{\alpha}\frac{1}{f(p)}\frac{\partial^{|\alpha|} f}{\partial x^\alpha}(p)\left(\frac{w}{2}\right)^k\right|\leq 1 .
$$
\end{corollary}

\begin{proof}
Observe that by the triangle inequality,
$$
\left|\sum_{k=1}^{\deg(f)}\sum_{|\alpha|=k}\frac{1}{k!}\binom{k}{\alpha}\frac{1}{f(p)}\frac{\partial^{|\alpha|} f}{\partial x^\alpha}\left(\frac{w}{2}\right)^k\right|\leq  \sum_{k=1}^{\deg(f)}\sum_{|\alpha|=k}\left|\frac{1}{k!}\binom{k}{\alpha}\frac{1}{f(p)}\frac{\partial^{|\alpha|} f}{\partial x^\alpha}\left(\frac{w}{2}\right)^k\right|.
$$
Now, we can substitute the bound on the derivatives in Proposition~\ref{A:proposition:derivativebound} as well as the assumed bound on $w$, resulting in an upper bound of
\begin{equation*}
\sum_{k=1}^{\deg(f)}\sum_{|\alpha|=k}\frac{1}{k!}2^{(n-1)(k+1)}\left(\frac{\ln(1+2^{2-2n})}{2^n}\right)^k.
\end{equation*}

Since there are $\binom{n+k-1}{k}$ possibilities for $\alpha$ when $|\alpha|=k$, which can be trivially bounded from above by $2^{n+k-1}$, we can bound the expression above by
\begin{align*} 
  2^{2n-2}\sum_{k=1}^{\deg(f)}\frac{1}{k!}(\ln(1+2^{2-2n}))^k.
\end{align*}
Since this sum is a truncated version of the Taylor series expansion of $e^x-1$ centered at $0$ with $x=\ln(1+2^{2-2n})$, this sum is bounded above by $2^{2-2n}$, and, hence, the entire expression is bounded above by $1$.
\end{proof}

Using Corollary~\ref{A:corollary:testbound}, we can develop bounds on the size of a region which guarantees the success of the given predicate.  We make this explicit in the following corollary:

\begin{corollary}\label{corollary:C0:initial}
Let $f\in\mathbb{R}[x_1,\dots,x_n]$ and $J\subseteq\mathbb{R}^n$ an $n$-dimensional cube with midpoint $m=m(J)$ and width $w=w(J)$ such that
$w\leq\frac{\dist_{\CC}(m,f)\ln(1+2^{2-2n})}{2^{n-1}\deg(f)}$.
Then, $C(J)$ is true.
\end{corollary}
\begin{proof}
From Equation~(\ref{eq:simplified:centered:form}), we see that $0\not\in\square f(J)$ is equivalent to 
$$
\left(\sum_{|\alpha|=1}^d \frac{\left|\partial^\alpha f(m)\right|}{|f(m)|\alpha!}\left(\frac{w}{2}\right)^{|\alpha|}\right)<1.
$$
This inequality arises because, in Equation~(\ref{eq:simplified:centered:form}), $\square f(J)$ is an interval centered at the origin shifted by $f(m)$.  In order for $0$ to be excluded from this interval, the shift by $f(m)$ must be larger than the half-width of the interval.  Dividing both sides by $|f(m)|$, we get exactly the expression in Corollary~\ref{A:corollary:testbound}.
\end{proof}

While Corollary~\ref{corollary:C0:initial} gives a test for $C(J)=\true$ for an $n$-dimensional cube, this test is not enough to prove that $C$ is a distance-diameter test because both sides of the inequality involve region $J$.  In particular, the midpoint of the region $J$ appears on the right-hand-side of the inequality.  The following lemma changes the right-hand-side of the inequality to depend on any point within $J$, instead of the midpoint.

\begin{lemma}\label{lemma:changefrommidpoint}
Let $f\in\mathbb{R}[x_1,\dots,x_n]$ and $J\subseteq\mathbb{R}^n$ an $n$-dimensional cube with midpoint $m=m(J)$ and width $w=w(J)$.  Suppose that $x\in J$ and $C$ and $k$ are positive constants.  If
$$
w\leq\frac{2C\dist_{\CC}(x,h)}{2k+C\sqrt{n}},
$$
then
$$
w\leq \frac{C \dist_{\CC}(m,h)}{k}.
$$
\end{lemma}
\begin{proof} We follow the ideas of the argument in~\cite{BurrKrahmer:2012}.  We observe that 
  \begin{equation}\label{equation:changefrommidpoint}
      w =\left(1+\frac{C\sqrt{n}}{2k}\right)w-\frac{C\sqrt{n}}{2k}w \leq\frac{C \dist_{\CC}(x,h)}{k}-\frac{C\sqrt{n}}{2k}w=\frac{C}{k}\left(\dist_{\CC}(x,h)-\frac{\sqrt{n}}{2}w\right),
\end{equation}
where the inequality follows from the assumed upper bound on $w$.  Suppose that $\alpha$ is the closest point of $V(h)$ to $m$, then, by the triangle inequality, $\dist_{\CC}(m,h)\geq \dist_{\CC}(x,\alpha)-\dist_{\CC}(x,m)$.  The distance $\dist_{\CC}(x,m)$ is at most the radius of $J$, which is $\frac{\sqrt{n}}{2}w$.  Moreover, the closest point on $V(h)$ to $x$ is distance at most the distance to $\alpha$, so $\dist_{\CC}(x,\alpha)\geq \dist_{\CC}(x,h)$.  Hence, $\dist_{\CC}(m,h)\geq \dist_{\CC}(x,h)-\frac{\sqrt{n}}{2}w$.  By substituting this into Expression~(\ref{equation:changefrommidpoint}), the desired result follows.
\end{proof}

By combining Corollary~\ref{corollary:C0:initial} with Lemma~\ref{lemma:changefrommidpoint}, we explicitly show that the predicate $C$ is a distance-diameter test:

\begin{corollary}\label{corollary:C0:localsizebound}
Let $f\in\mathbb{R}[x_1,\dots,x_n]$ and $J\subseteq\mathbb{R}^n$ an $n$-dimensional cube with width $w=w(J)$.  If there is a point $x\in J$ such that
$$
w\leq\frac{2\ln\left(1+2^{2-2n}\right)\dist_{\CC}(x,f)}{2^n\deg(f)+\sqrt{n}\ln\left(1+2^{2-2n}\right)},
$$
then, $C(J)$ is true.
\end{corollary}
\begin{proof}
This result follows from Corollary~\ref{corollary:C0:initial} with Lemma~\ref{lemma:changefrommidpoint} by letting $C=\ln\left(1+2^{2-2n}\right)$ and $k=2^{n-1}\deg(f)$.
\end{proof}

Since the diameter of an $n$-dimensional cube whose side is of length $w$ is scaled by $\sqrt{n}$, we have the following corollary:

\begin{corollary}\label{corollary:C0:localsizebound:diameter}
Let $f\in\mathbb{R}[x_1,\dots,x_n]$ and $J\subseteq\mathbb{R}^n$ an $n$-dimensional cube.  If there is a point $x\in J$ such that
$$
\diam(J)\leq\frac{2\sqrt{n}\ln\left(1+2^{2-2n}\right)\dist_{\CC}(x,f)}{2^n\deg(f)+\sqrt{n}\ln\left(1+2^{2-2n}\right)}.
$$
Then, $C(J)$ is true.  Therefore, $C$ is a diameter-distance test.
\end{corollary}

\section{The Modified Plantinga-Vegter Algorithm}
\label{sec:PV-alg}

In this section, we provide a modified form of the \PV\ algorithm~\cite{Plantinga:2004} for curve and surface approximation.  Our version of the algorithm uses a slightly stronger variant of their, but makes both tests of the appropriate form for the application of Corollary~\ref{corollary:C0:localsizebound:diameter}.  We begin by reviewing the original \PV\ algorithm and then discuss our generalization and adaptation.  

\subsection{The Original \PV\ Algorithm}

Let $f\in\mathbb{R}[x,y]$ or $\mathbb{R}[x,y,z]$ be a square-free polynomial such that its real zero set $V_{\RR}(f)$ is smooth (see Footnote~\ref{fn:correctness} for a brief discussion of the requirement for the curve to be bounded for correctness of the approximation).  The \PV\ algorithm recursively subdivides an initial input square or cube $I$ with a quad-tree or oct-tree data structure until at least one of the following tests holds on each subregion $J$ (in the literature, these tests are often referred to as $C_0$ and $C_1$):
\begin{align*}
C_0(J)&=\true\quad\text{if and only if}\quad 0\not\in\square f(J)\\
C_1(J)&=\true\quad\text{if and only if}\quad 0\not\in\langle \square\nabla f(J),\square\nabla f(J)\rangle.
\end{align*}
For the purposes of curve approximation, when $C_0(J)$ is $\true$, the variety does not enter the region $J$, and so $J$ can be discarded.  On the other hand, when $C_1(J)$ holds, the curve or surface does not bend much within the region $J$.

The \PV\ algorithm is an instance of an Abstract Subdivision-based Algorithm that uses bisection and the two tests $C_0$ and $C_1$, see Algorithm~\ref{algorithm:abstract}.  We explicitly include the algorithm here for completeness and to illustrate the simplicity of the approach.  Given an input polynomial $f$ and region $I$, the \PV\ algorithm constructs a partition $P$ of $I$ so that for all regions $J$ of the partition, either $C_0(J)=\true$ or $C_1(J)=\true$.  Initially, $P=\{I\}$.
\begin{algorithm1}{Main subdivision of \PV\ algorithm}\label{algorithm:PV}
\noindent Repeatedly subdivide (into $4$ or $8$ children) each $J\in P$ until one of the following conditions hold:\\[.15cm]
\indent $C_0(J)$ is $\true$ or $C_1(J)$ is $\true$.
\end{algorithm1}

After every sub-region $J$ satisfies $C_0(J)$ or $C_1(J)$, the authors of~\cite{Plantinga:2004} perform post-processing steps, which include balancing the tree, evaluating the sign of $f$ on the corners of each $J$ in $P$, and using sign changes along the sides of regions $J$ to detect and approximate the curve or surface.  Their approximation is topologically correct for bounded curves as there is an ambient isotopy between the approximation and the variety $V_{\RR}(f)$.  Additionally, by further subdivision, the isotopy can be made sufficiently small so that the Hausdorff distance between the approximation and the variety is as small as desired.  We note that it is possible to extend the \PV\ algorithm in the plane to provide a topologically correct approximation even when $V_{\RR}(f)$ is unbounded, when $V_{\RR}(f)$ is singular, and when $I$ is not a bounding box, see~\cite{Burr:Isotopic}.  In this paper, however, we focus on the original \PV\ algorithm, but without the restriction of a bounded curve.

Our main target is to compute the number of regions that the \PV\
algorithms construct, and not to approximate the curve or
surface, per se.  Therefore, we focus, exclusively, on the $C_0$ and
$C_1$ tests and apply them in arbitrary dimensions.  More precisely,
let $f\in\mathbb{R}[x_1,\dots,x_n]$ be such that its real zero set
$V_{\RR}(f)$ is smooth.  Let $I\subseteq\mathbb{R}^n$ be an
$n$-dimensional real cube.  Then, we can generalize the tests $C_0$
and $C_1$, along with Algorithm~\ref{algorithm:PV}, to $n$ dimensions,
where the subdivision splits an $n$-dimensional cube into $2^n$
children.  We mention that, in this case, we no longer use the output
of the algorithm to construct an approximation to $V_{\RR}(f)$.

\subsection{Modifying the \texorpdfstring{$C_1$}{C1} test}
\label{sec:extend-C1}

As presented above, the $C_0$ test is of the form considered in Corollary~\ref{corollary:C0:localsizebound:diameter}, so it is a diameter-distance test.  On the other hand, the $C_1$ test is not of this form; therefore, it is not clear if the $C_1$ test is a diameter-distance test.  The difficulty in applying the corollary in this case is that arithmetic operations are performed on intervals after an application of interval methods.  In this section, we describe an alternate $C_1$ test that satisfies the assumptions of Corollary~\ref{corollary:C0:localsizebound:diameter}.

The predicate $C_1(J)$ has the following two consequences that are fundamental in the proof of the correctness of the \PV\ algorithm in~\cite{Plantinga:2004}:
\begin{enumerate}
\item If a region $J$ satisfies the $C_1$ condition, then, in $J$, there cannot be any pair of gradient vectors of $f$ which are orthogonal to each other.\label{fact1:PVtest}
\item The variety $V_{\RR}(f)$ is parametrizable in the direction of at least one of the coordinate axes. \label{fact2:PVtest}
\end{enumerate} 
Fact~\ref{fact2:PVtest} is a direct consequence of Fact~\ref{fact1:PVtest}, but it is used so frequently in the proofs in~\cite{Plantinga:2004}, that it is worthwhile to mention it
explicitly.

We now modify the $C_1$ test in arbitrary dimensions so that it has the form in the assumptions of Corollary~\ref{corollary:C0:localsizebound:diameter}.  Consider the function $g:\mathbb{R}^n\times\mathbb{R}^n\rightarrow\mathbb{R}$, defined as 
$$
g(x_1,\dots,x_n,y_1,\dots,y_n)=\langle\nabla f(x_1,\dots,x_n),\nabla f(y_1,\dots,y_n)\rangle.
$$
It follows that, if, for a region $J$, $0\not\in\square g(J\times J)$, then there is no pair of gradient vectors of $f$ in $J$ which are orthogonal to each other.  Thus the modified $C_1$ test, briefly denoted $C_1'$, is as follows:
$$
C_1'(J)=\true\quad\text{if and only if}\quad0\not\in\square g(J\times J).
$$
Therefore, when $C_1'(J)$ is true, we can conclude the truth of Facts~\ref{fact1:PVtest} and~\ref{fact2:PVtest}, and the application of an interval method appears as the last step as opposed to in an intermediate step.  Therefore, $C_1'$ satisfies the assumptions in Corollary~\ref{corollary:C0:localsizebound:diameter}, and, therefore, is a diameter-distance test.  For the rest of the paper, all references to the $C_1$ test refer to this new $C_1'$ test.  In particular, all discussions of the \PV\ algorithm refer to the modified \PV\ algorithm.

Let $J$ be an $n$-dimensional real cube with midpoint $m=m(J)$ and side length $w=w(J)$.  The explicit formula for the $C_0$ test appears in the proof of Corollary~\ref{corollary:C0:initial}.  Since the $C_1$ test is based on the function $g$ whose domain is $2n$-dimensional and the square $J\times J$ has midpoint $(m,m)$, but side length $w$, the $C_1$ test simplifies, in terms of multi-index notation, to 
\begin{equation}\tag{$C_1$}
\sum_{|\alpha|+|\beta|\geq 1}\left|\sum_{i=1}^n \frac{\partial^{\alpha+e_i}f(m)\partial^{\beta+e_i}f(m)}{\|\nabla f(m)\|^2\alpha!\beta!}\right|\left(\frac{w}{2}\right)^{|\alpha|+|\beta|}<1,
\end{equation}
where $e_i$ is the $i$-th standard basis vector.  We additionally note that since $f$ and $J$ are real, $\|\nabla f(m)\|^2=g(m,m)$.  Additionally, for future reference, we collect and adapt the statement of Corollary~\ref{corollary:C0:localsizebound} to the case of $g$ on the region $J\times J$ in the following corollary:
\begin{corollary}\label{corollary:C1:localsizebound}
Let $f\in\mathbb{R}[x_1,\dots,x_n]$ and $g\in\mathbb[x_1,\dots,x_n,y_1,\dots,y_n]$ be $g(x_1,\dots,x_n,y_1,\dots,y_n)=\langle\nabla f(x_1,\dots,x_n),\nabla f(y_1,\dots,y_n)\rangle$.  Let $J\subseteq\mathbb{R}^n$, and suppose that there is a point $(a,b)\in J\times J$ such that 
$$w\leq\frac{2\ln\left(1+2^{2-4n}\right)\dist_{\CC}((a,b),g)}{2^{2n+1}(\deg(f)-1)+\sqrt{2n}\ln\left(1+2^{2-4n}\right)}.$$  Then, $C_1(J)$ is true.
\end{corollary}

We end this section by collecting a corollary of Corollary~\ref{corollary:C1:localsizebound} which resembles a diameter-distance test and will be used in the next section: 
\begin{corollary}\label{corollary:C1:localsizebound:diameterdistance}
Let $f\in\mathbb{R}[x_1,\dots,x_n]$ and $g\in\mathbb[x_1,\dots,x_n,y_1,\dots,y_n]$ be $g(x_1,\dots,x_n,y_1,\dots,y_n)=\langle\nabla f(x_1,\dots,x_n),\nabla f(y_1,\dots,y_n)\rangle$.  Let $J\subseteq\mathbb{R}^n$, and suppose that there is a point $x\in J$ such that 
$$\diam(J)\leq\frac{2\sqrt{n}\ln\left(1+2^{2-4n}\right)\dist_{\CC}((x,x),g)}{2^{2n+1}(\deg(f)-1)+\sqrt{2n}\ln\left(1+2^{2-4n}\right)}.$$  Then, $C_1(J)$ is true.
\end{corollary}

\section{Worst-Case Bounds}
\label{sec:worst-case-bounds}

In this section, we provide worst-case complexity bounds for the
modified \PV\ algorithm.  We bound both the number of regions produced
by the subdivision as well as the overall bit-complexity of the
algorithm.  In the next section, we give examples which show that
these bounds are tight in the worst case.

\subsection{Non-adaptive Bounds}
In this section, we use Proposition~\ref{prop:global-size-bound} to bound the number of regions produced by the \PV\ algorithm.  We assume\footnote{The argument in this section can be directly generalized for $f\in\mathbb{Q}[x_1,\dots,x_n]$ and $I$ whose corners are in $\mathbb{Q}^n$.  We leave the details to the interested reader.} that $f\in\ZZ[x_1,\dots,x_n]$ and the fixed initial input region $I$ has corners in $\mathbb{Z}^n$.  Suppose that we can find a $\delta$ so that
$$
0<\delta\leq\min_{x\in I}\max\left\{\dist_{\CC}(x,f),\dist_{\CC}((x,x),g)\right\},
$$  
and define 
$$
K=\min \left\{\frac{2\sqrt{n}\ln\left(1+2^{2-2n}\right)}{2^n\deg(f)+\sqrt{n}\ln\left(1+2^{2-2n}\right)},
\frac{2\sqrt{n}\ln\left(1+2^{2-4n}\right)}{2^{2n+1}(\deg(f)-1)+\sqrt{2n}\ln\left(1+2^{2-4n}\right)} \right\}.
$$

We observe that the terms in the definition of $K$ are the
coefficients in Corollaries~\ref{corollary:C0:localsizebound} and
\ref{corollary:C1:localsizebound:diameterdistance}.  With a slight
modification to Proposition~\ref{prop:global-size-bound}, we can
substitute $K$ and $\delta$ from above into
Equation~(\ref{eq:nonadaptive:bound}) to get the following corollary:

\begin{corollary}
  \label{corollary:separationboundcurvescomplexity}
Let $f\in\mathbb{R}[x_1,\dots,x_n]$.  Then, the \PV\ algorithm performs at most $\max\left\{1,\left(\frac{2\diam(I)}{K\delta}\right)^n\right\}$ subdivisions.
\end{corollary}

We spend the remainder of this section computing a lower bound for $\delta$.  We begin by observing that $f$ is a polynomial in $n$ variables and $g$ is a polynomial in $2n$ variables.  In other words, the varieties $V_{\CC}(f)$ and $V_{\CC}(g)$ are embedded in two different spaces.  It becomes easier to study and compare the varieties if they are subsets of the same space; therefore, we consider the image of $V_\CC(f)$ in the diagonal of a $2n$-dimensional space.  In particular, let the variables of $\CC^{2n}$ be $\{x_1,\dots,x_n,y_1,\dots,y_n\}$.  The diagonal $\Delta$ consists of all points of the form $x_i=y_i$; then, $\Delta$ is $n$-dimensional and we can identify $\CC^n$ with $\Delta$.  In our case, we write $f^\Delta$ for the polynomial system $f(x_1,\dots,x_n)$ and $x_i-y_i$ for all $i$.  We note that for $x\in\CC^n$, $\dist_\CC((x,x),f^\Delta)=\sqrt{2}\,\dist_\CC(x,f)$.  Therefore, we are interested in computing a lower bound for 
\begin{multline}\label{eq:2nlowerbound}
\min_{x\in I}\max\left\{\frac{1}{\sqrt{2}}\dist_{\CC}((x,x),f^\Delta),\dist_{\CC}((x,x),g)\right\}\\
\geq \frac{1}{\sqrt{2}}\min_{x\in I}\max\left\{\dist_{\CC}((x,x),f^\Delta),\dist_{\CC}((x,x),g)\right\}.
\end{multline}
We now focus on computing a lower bound on the RHS of Inequality~(\ref{eq:2nlowerbound}).

First we introduce some notation.  Let $I^\Delta$ be the image of $I$
in $\Delta$, i.e., $I^\Delta=\{(x,x)\in\Delta:x\in I\}$.  Moreover,
let
$C_\varepsilon=\left([-\varepsilon,\varepsilon]\times[-i \varepsilon,i
  \varepsilon]\right)^{2n}$ be the cube of side length $2\varepsilon$
centered at the origin in $\CC^{2n}$.  Then, we write
$I^\Delta_\varepsilon=I^\Delta\oplus C_\varepsilon$, where $\oplus$
denotes the Minkowski sum.  We observe that for all
$(x,y)\in I^\Delta_\varepsilon$, the distance from $(x,y)$ to
$I^\Delta$ is at most $2\sqrt{n}\varepsilon$
since that is the largest
distance from a point in $C_\varepsilon$ to the origin.  Similarly, if
$(x,y)\in\CC^{2n}$ is not in $I^\Delta_\varepsilon$, then the distance
from $(x,y)$ to $I^\Delta$ is more than $\varepsilon$ since
$C_\varepsilon$ contains the closed ball of radius $\varepsilon$ centered at
the origin.

Suppose that we can find a positive integer $k$ so that for any $(x,x)\in V_{\CC}(f^\Delta,g)$, the distance between $(x,x)$ and $I^\Delta$ is at least $\frac{\sqrt{n}}{2^{k-1}}$.  Then, we may use a bound of~\cite{JPT:SemiAlgebraic} to find a lower bound for the RHS of Inequality~(\ref{eq:2nlowerbound}) as follows:
\begin{proposition}\label{prop:separationbound}
Let $f\in\mathbb{Z}[x_1,\dots,x_n]$ be of degree $d$ and $g\in\mathbb{R}[x_1,\dots,x_n,y_1,\dots,y_n]$ be $g(x_1,\dots,x_n,y_1,\dots,y_n)= \langle\nabla f(x_1,\dots,x_n),\nabla f(y_1,\dots,y_n)\rangle$.  Suppose that $I\subseteq\mathbb{R}^n$ is an axis-aligned $n$-dimensional cube whose corners have integral coordinates.  Let $H$ be the maximum absolute value of the coefficients of $f$ and coordinates of the corners of $I$.  Suppose that $f^\Delta=\{f,x_i-y_i\}$ is the polynomial system corresponding to the image of $V_\CC(f)$ in the diagonal of $\mathbb{C}^{2n}$ and $I^\Delta=\{(x,x):x\in I\}$ is the image of $I$ in the diagonal of $\mathbb{C}^{2n}$.  Let $k$ be a positive integer so that for any $(x,x)\in V_{\CC}(f^\Delta,g)$, the distance between $(x,x)$ and $I^\Delta$ is more than $\frac{\sqrt{n}}{2^{k-1}}$.  Then, 
\begin{multline}
\min_{x\in I}\max\left\{\dist_{\CC}((x,x),f^\Delta),\dist_{\CC}((x,x),g)\right\}\\\geq\frac{1}{2^{k+1}}\left(2^{4-4n}\max\left\{2^{(2d-2)(k+1)}nd^2H^2,60n+8\right\}(2d-2)^{8n}\right)^{-4n2^{8n}(2d-2)^{8n}}.
\end{multline}
\end{proposition}

\begin{proof}
  If $d=1$, then $g$ is a nonzero constant, so the bound holds
  trivially.  Therefore, we assume that $d\geq 2$.  We observe that
  $g$ has degree $2d-2$, and, since the maximum absolute value of the
  coefficients of $\nabla f$ is $dH$, the maximum absolute value of
  the coefficients of $g$ is $nd^2H^2$.

  Let $\varepsilon=\frac{1}{2^k}$; by the assumption on $k$, it
  follows that $V_{\mathbb{C}}(f^\Delta,g)\cap I^\Delta_\varepsilon$
  is empty.  We proceed by applying a homothety centered at the origin
  by a factor of $2^k$ in $\mathbb{C}^{2n}$.  Therefore, if we let
  $\widetilde{I}^\Delta$ be the image of $I^\Delta$ after applying the
  homothety, then applying the homothety to $I^\Delta_\varepsilon$
  results in $\widetilde{I}^\Delta_1$.  Let $\widetilde{f}$ and
  $\widetilde{g}$ be the images of $f$ and $g$ under the homothety and
  a suitable scaling to restore integer coefficients.  Then, the
  maximum absolute value of the coefficients of $\widetilde{f}$ is
  $2^{dk}H$ and the maximum absolute value of the corners of
  $\widetilde{I}^\Delta$ is $2^{k}H$.  Additionally, the linear terms
  of the form $x_i-y_i$ are unchanged and the maximum absolute value
  of the coefficients of $g$ is $2^{(2d-2)k}nd^2H^2$.

  Next, we identify $\mathbb{C}^{2n}$ with $\mathbb{R}^{4n}$ by
  decomposing each complex variable into two real variables.  This
  doubles the number of polynomials and scales the maximum absolute
  value of the coefficients by binomial coefficients, which can be
  trivially bounded by $2^d$ for the polynomials coming from
  $\widetilde{f}$ and $2^{2d-2}$ for the polynomials coming from
  $\widetilde{g}$.  Hence, the maximum absolute value of the
  coefficients coming from $\widetilde{f}$ is at most $2^{d(k+1)}H$,
  and the maximum absolute value of the coefficients coming from
  $\widetilde{g}$ is at most $2^{(2d-2)(k+1)}nd^2H^2$.

  We observe that if $I=\prod[a_i,b_i]$, then $\widetilde{I}^\Delta_1$
  can be defined by the inequalities:
  \begin{align*}
    2^ka_i-1\leq \Re(x_i)&,\Re(y_i)\leq 2^kb_i+1\\
    -1\leq \Im(x_i)&,\Im(y_i)\leq 1\\
    -2\leq\Re(x_i)&-\Re(y_i)\leq 2.
  \end{align*}
  This system accounts for $10n$ inequalities with largest absolute
  value of the coefficients at most $2^{k}H+1$.  Moreover,
  $\widetilde{f}^\Delta$ corresponds to $2n+2$ equalities while
  $\widetilde{g}$ corresponds to $2$ equalities.  By applying
 ~\cite[Theorem 1.2]{JPT:SemiAlgebraic}, we get that the distance
  between $V_\mathbb{C}(\widetilde{f}^\Delta)$ and
  $V_\mathbb{C}(\widetilde{g})$ within $\widetilde{I}_1^\Delta$ is at
  least
  $$  \left(2^{4-4n}\max\left\{2^{(2d-2)(k+1)}nd^2H^2,60n+8\right\}(2d-2)^{8n}\right)^{-4n2^{8n}(2d-2)^{8n}}.
  $$
  By scaling this by $\frac{1}{2^k}$ to remove the homothety and
  appealing to the triangle inequality, we get the desired result.
\end{proof}

In the remainder of this section, we find an upper bound for $k$.  We find this bound by computing a separation bound between $I^\Delta$ and $V_\CC(f^\Delta,g)$.

\begin{proposition}\label{prop:separationbound2}
Let $f\in\mathbb{Z}[x_1,\dots,x_n]$ be smooth and of degree $d$ and $g\in\mathbb{R}[x_1,\dots,x_n,y_1,\dots,y_n]$ be $g(x_1,\dots,x_n,y_1,\dots,y_n)= \langle\nabla f(x_1,\dots,x_n),\nabla f(y_1,\dots,y_n)\rangle$.  Suppose that $I\subseteq\mathbb{R}^n$ is an axis-aligned $n$-dimensional cube whose corners have integral coordinates.  Let $H$ be the maximum absolute value of the coefficients of $f$ and coordinates of the corners of $I$.  Suppose that $f^\Delta=\{f,x_i-y_i\}$ is the polynomial system corresponding to the image of $V_\CC(f)$ in the diagonal of $\mathbb{C}^{2n}$ and $I^\Delta=\{(x,x):x\in I\}$ is the image of $I$ in the diagonal of $\mathbb{C}^{2n}$.  Let $(x,x)\in V_{\CC}(f^\Delta,g)$, then the distance between $(x,x)$ and $I^\Delta$ is at least
$$
\left(2^{4-2n}\max\left\{2^{(2d-2)}nd^2H^2,32n+8\right\}(2d-2)^{4n}\right)^{-2n2^{4n}(2d-2)^{4n}}.
$$
\end{proposition}

\begin{proof}
If $d=1$, then $g$ is a nonzero constant, so the bound holds vacuously.  Therefore, we assume that $d\geq 2$.  Throughout this proof, we restrict our attention to $I^\Delta_1$ and we observe that if $(x,x)\in V_{\CC}(f^\Delta,g)$ is outside of $I^\Delta_1$, then $1$ is a lower bound on its distance to $I^\Delta$.  Since $f$ is smooth, it follows that $V_\CC(f^\Delta,g)\cap\mathbb{R}^{2n}$ is empty.  Therefore, by a compactness argument, $V_\CC(f^\Delta,g)\cap I^\Delta_1$ is bounded away from the real points in the diagonal, i.e., $\mathbb{R}^{2n}\cap\Delta$.  Moreover, since $V_\CC(f^\Delta,g)$ contains no real points, it follows that $2\sum\Im(x_i)^2$ is bounded away from zero for all $(x,x)\in V_\CC(f^\Delta,g)\cap I_1^\Delta$ and the sum is a lower bound on the square of the distance to $I_\Delta^1$.  We now proceed to find a lower bound on this sum.

As in Proposition~\ref{prop:separationbound}, we identify $\mathbb{C}^{2n}$ with $\mathbb{R}^{4n}$.  Since no homotheties are required, $I^\Delta_1$ corresponds to $10n$ inequalities with maximum coefficient size $H+1$, $f^\Delta$ corresponds to $2n+2$ equalities with coefficient size at most $2^dH$, and $g$ corresponds to $2$ equalities with coefficient size at most $2^{(2d-2)}nd^2H^2$.  Finally, the sum of interest is of degree $2$ with maximum coefficient size of $2$.  By applying~\cite[Theorem 1.1]{JPT:SemiAlgebraic}, we get that on $V_\CC(f^\Delta,g)\cap I_1^\Delta$, the sum $2\sum\Im(x_i)^2$ is at least
$$
\left(2^{4-2n}\max\left\{2^{(2d-2)}nd^2H^2,32n+8\right\}(2d-2)^{4n}\right)^{-4n2^{4n}(2d-2)^{4n}}.
$$
Since the sum $2\sum\Im(x_i)^2$ is the square of the distance from $(x,x)$ to $\mathbb{R}^{2n}$, which contains $I^\Delta_1$, so, by taking the square root, the result follows.
\end{proof}

By combining Corollary~\ref{corollary:separationboundcurvescomplexity} with Propositions~\ref{prop:separationbound} and~\ref{prop:separationbound2}, we obtain an explicit bound for the number of terminal regions produced by the \PV\ algorithm.

\begin{theorem}\label{thm:PV:steps}
Let $f\in\mathbb{Z}[x_1,\dots,x_n]$ be smooth and of degree $d$ and $I\subseteq\mathbb{R}^n$ be an axis-aligned $n$-dimensional cube whose corners have integral coordinates.  Let $H$ be the maximum absolute value of the coefficients of $f$ and coordinates of the corners of $I$.  The number of (terminal) regions produced by the \PV\ algorithm is
$$
2^{O(n^32^{24n}d^{12n+1}(d+\lg H+n\lg d))}.
$$
\end{theorem}

\begin{proof}
We observe that for any positive constant $a$, $\lg(\ln(1+2^{2-an}))=O(-n)$ and $\sqrt{x}\ln(1+2^{2-2x})$ is bounded, so $-\lg K=O(n+\lg d)$.  By Proposition \ref{prop:separationbound}, $-\lg \delta=O(n2^{16n}d^{8n}(dk+\lg H+n\lg d))$.  Next, since $k$ is an integer and an the exponent of $2$, $k$ can be chosen to be within $1$ of the base 2 logarithm of the bound in Proposition~\ref{prop:separationbound2}.  Therefore, $k=O(n2^{8n}d^{4n}(d+\lg H+n\lg d))$.  Substituting this into the bound for $\delta$, we find that $-\lg \delta = O(n^22^{24n}d^{12n+1}(d+\lg H+n\lg d))$.  Substituting the bounds into the expression in Corollary~\ref{corollary:separationboundcurvescomplexity} results in the stated complexity.
\end{proof}

\subsection{Adaptive Bounds}
\label{sec:adaptive-bounds}

In this section, we use continuous amortization to adaptively compute the number of boxes created by the \PV\ algorithm.  We follow the formulation of continuous amortization in Theorem~\ref{thm:local-sz-bd}.  While Corollary~\ref{corollary:C0:localsizebound:diameter} shows that the $C_0$ test can be substituted directly into the integral of Proposition~\ref{prop:adaptive-size-bound}, the $C_1$ test is slightly more challenging to use, even with Corollary~\ref{corollary:C1:localsizebound:diameterdistance} in hand, since it involves both $n$-dimensional and $2n$-dimensional spaces.  We, therefore, return to the original formulation of continuous amortization in Theorem~\ref{thm:local-sz-bd}.  We observe that Corollary~\ref{corollary:C0:localsizebound} can be reformulated into a local size bound since the volume of an $n$-dimensional cube is the width of the cube to the $n^{\text{th}}$ power, namely, 
$$
G_0(x)=\left(\frac{2\ln\left(1+2^{2-2n}\right)\dist_{\CC}(x,f)}{2^n\deg(f)+\sqrt{n}\ln\left(1+2^{2-2n}\right)}\right)^n
$$
is a local size bound for the $C_0$ test.  Similarly, Corollary~\ref{corollary:C1:localsizebound} can be reformulated into a local size bound as follows:
$$
G_1(x)=\left(\frac{2\ln\left(1+2^{2-4n}\right)\dist_{\CC}((x,x),g)}{2^{2n+1}(\deg(f)-1)+\sqrt{2n}\ln\left(1+2^{2-4n}\right)}\right)^n.
$$
In this case, even though the test in Corollary~\ref{corollary:C1:localsizebound} uses points $(a,b)\in J\times J$, since the statement is existential, the upper bound only gets smaller when restricted to the points in $J\times J$ on the diagonal $\Delta$.  Applying these local size bounds to Theorem~\ref{thm:local-sz-bd} gives the following result:

\begin{theorem}\label{theorem:PV:CA:complexity}
Let $f\in\mathbb{R}[x_1,\dots,x_n]$ and $g\in\mathbb{R}[x_1,\dots,x_n,y_1,\dots,y_n]$ be $g(x_1,\dots,x_n,y_1,\dots,y_n)= \langle\nabla f(x_1,\dots,x_n),\nabla f(y_1,\dots,y_n)\rangle$.  Suppose that $I\subseteq\mathbb{R}^n$ is an axis-aligned $n$-dimensional cube.  The number of (terminal) regions after the subdivision performed by the \PV\ algorithm (before balancing) is bounded above by the maximum of $1$ and
\begin{align*}
  2^n\int_I\min \left\{  \left(\frac{2^n\deg(f)+\sqrt{n}\ln\left(1+2^{2-2n}\right)}{2\ln\left(1+2^{2-2n}\right)\dist_{\CC}(x,f)}\right)^n,\left(\frac{2^{2n+1}(\deg(f)-1)+\sqrt{2n}\ln\left(1+2^{2-4n}\right)}{2\ln\left(1+2^{2-4n}\right)\dist_{\CC}((x,x),g)}\right)^n
  \right\} \, dV_n
\end{align*}
where $dV_n$ is the $n$-dimensional volume form.  Moreover, the algorithm does not terminate if and only if the integral diverges.
\end{theorem}
\begin{proof}
This is a straight-forward application of continuous amortization from Theorem~\ref{thm:local-sz-bd} with $\varepsilon_1=2^{-n}$.  The only statement left to prove is that if the integral diverges, then the algorithm does not terminate.  The integral diverges if and only if there exists a point $x\in I$ so that $\dist_\CC(x,f)=0$ and $\dist_\CC((x,x),g)=0$.  This, however, only happens when $f$ has a real singularity, and regions containing real singularities never pass either of the $C_0$ or $C_1$ tests.
\end{proof}

This integral provides a more adaptive and accurate estimate on the complexity than the worst-case {\em a priori} bounds based on the size of the input because it does not assume that the worst case occurs at every point (or even at any point).  Moreover, this integral can be evaluated even when the input polynomial has complex (but not real) singularities.  Additionally, this integral applies even when $f$ does not have integral coefficients.

\subsection{Overall Bit-complexity Bound}

In this section, we extend the results of Theorems~\ref{thm:PV:steps} and~\ref{theorem:PV:CA:complexity} to bound the bit-complexity of the \PV\ algorithm using both adaptive and non-adaptive approaches.  We begin by bounding the cost for evaluating each of the tests $C_0$ and $C_1$ on an arbitrary $n$-dimensional cube.  In this section, we use $O(\cdot)$ and $O_B(\cdot)$ to denote the arithmetic complexity and bit-complexity, respectively.  The soft-$O$ notation, $\sO(\cdot)$ and $\sOB(\cdot)$, mean that we are ignoring logarithmic factors of the dominant term.

A closer look at the predicates $C_0$ and $C_1$ and the centered form (see Section~\ref{Section:localsizebound}) reveals that each step of the \PV\ algorithm consists of a multivariate Taylor shift.  In particular, given a polynomial $F \in \mathbb{Z}[x_1, \dots, x_n]$ and dyadic rational numbers $a_1, \dots, a_n$, we recursively compute the coefficients of $F(x_1 + a_1, \dots, x_n + a_n)$, cf~\cite{mmt-tcs-2010}.

\begin{lemma}
\label{lem:biv-Taylor}
Consider a polynomial $F \in \ZZ[x_1,\cdots,x_n]$ of total degree $d$ and whose coefficients have maximum bit-size $\tau$, and integers $a_1,\dots,a_n$ of bit-size at most $\varrho$.  The Taylor shift $F(x_1 + a_1,\dots,x_n+a_n)$ costs $\sOB(d^{n+1}\varrho+d^n\tau)$.
\end{lemma}

\begin{proof}
We begin the proof with two observations: The maximum degree of any polynomial appearing in this proof is $d$ and the logarithm of the bit-size of the coefficients is $\widetilde{O}(d\varrho+\tau)$, see, e.g.,~\cite[Lemma 2.1]{GG-Taylor-97}.  We prove this lemma by induction; when $n=1$, this is a univariate Taylor shift, whose complexity is $\sOB(d^2\varrho+d\tau)$ by~\cite[Theorem 2.4]{GG-Taylor-97}.

For the inductive step, we assume that $d+1$ is a power of $2$.  We begin by calculating $(x_n+a_n)^{2^i}$ for $i=0,\dots,\lg d$.  Since each of these polynomials has coefficients of maximum bit-size $\sOB(d\varrho)$, and these can be computed through successive squaring, the total cost is $\sOB(d^2\varrho)$.  We now write 
$$
F(x_1+a_1,\dots,x_n+a_n)=F_0(x_1+a_1,\cdots,x_n+a_n)+(x_n+a_n)^{d/2}F_1(x_1+a_1,\cdots,x_n+a_n)
$$
where in each $F_i$, the degree in $x_n$ is at most $d/2$.  The cost to compute the product $(x_n+a_n)^{d/2}F_1(x_1+a_1,\cdots,x_n+a_n)$ is $\sOB(d^n(d\varrho+\tau))$.  By continuing this computation recursively, we see that the number of polynomials doubles each time and the maximum degree of $x_n$ halves each time, so the total cost of multiplication remains $\sOB(d^n(d\varrho+\tau))$ at every step.  The recursion has depth $\lg(d+1)$, and the final step of the recursion requires $(d+1)$ Taylor shifts on $(n-1)$ variables.  The result then follows from the inductive hypothesis.
\end{proof}

Using Theorem~\ref{thm:PV:steps} and Lemma~\ref{lem:biv-Taylor}, we can calculate the overall bit-complexity of the \PV\ algorithm.

\begin{theorem}\label{thm:nonadaptive-bit-complexity}
Let $f\in\mathbb{Z}[x_1,\dots,x_n]$ be smooth and of degree $d$ and $I\subseteq\mathbb{R}^n$ be an axis-aligned $n$-dimensional cube whose corners have integral coordinates.  Let $\tau=\lg H$ be the maximum bit-size of the coefficients of $f$ and the corners of $I$.  The overall bit-complexity of the \PV\ algorithm is 
$$
2^{O(n^32^{24n}d^{12n+1}(d+\tau+n\lg d))}\sOB(2^{26n}d^{14n+2}(d+\tau)).
$$
\end{theorem}

\begin{proof}
We observe that, after each subdivision in the \PV\ algorithm, the bit-size of the center of the Taylor shift increases by at most $1$.  To simplify the calculation, we charge each $n$-dimensional cube in the final partition for all intermediate $n$-dimensional cubes that contain it, proportionally to their relative areas.  Following the approach of~\cite[Section 7.1]{Burr:2016}, it follows that the total complexity cost of the \PV\ algorithm is at most twice the cost incurred by the terminal regions themselves.

We observe that the maximum bit-size of a Taylor shift is $O(-\lg\delta)$ from Theorem~\ref{thm:PV:steps}, so we replace $\varrho$ in Lemma~\ref{lem:biv-Taylor} by the bound from this theorem.  We also recall, from Proposition~\ref{prop:separationbound}, that $g$ is a polynomial of degree $2d-2$ in $2n$ variables whose coefficients have maximum bit-size $O(\tau+d+\lg n)$.  By substituting these values into Lemma~\ref{lem:biv-Taylor} and multiplying by the maximum number of regions, we arrive at the overall bit-complexity of
$$
2^{O(n^32^{24n}d^{12n+1}(d+\tau+n\lg d))}\sOB((2d-2)^{2n+1}(n^22^{24n}d^{12n+1}(d+\tau+n\lg d))+(2d-2)^{2n}(\tau+d+\lg n)),
$$
which simplifies to the desired expression.
\end{proof}

We observe that in the 2-dimensional case that frequently occurs in applications, the overall bit-complexity of the \PV\ algorithm is as follows:

\begin{corollary}
The bit-complexity of the \PV\ algorithm for curves is
$$
2^{O(d^{25}(d+\tau+\lg d))}\sOB(d^{30}(d+\tau)).
$$
\end{corollary}

We may also use Theorem~\ref{thm:local-sz-bd} to find an adaptive bound for the bit complexity.  To be able to use this Theorem, we need to define the appropriate functions $h_0$ and $h_1$ that compute the charges to the terminal regions depending on the $C_0$ and $C_1$ tests.  The main complexity costs in the $C_0$ and $C_1$ tests are the costs for the Taylor shifts.  Therefore, we use Lemma~\ref{lem:biv-Taylor} to derive appropriate cost functions.  We observe that for an $n$-dimensional cube $J$, the bit-size of the appropriate Taylor shift is at most $(\lg w(I)-\lg w(J))$.  By the discussion above, since the complexity cost of the \PV\ algorithm is at most twice the complexity cost of the terminal regions, we may focus on terminal regions.

If $J$ passes $C_0$, since $f$ is a degree $d$ polynomial in $n$ variables whose coefficients have maximum bit-size $\tau$, it follows that the charge associated to $J$ is $\sOB(d^{n+1}\lg w(I)-d^{n+1}\lg w(J)+d^n\tau)$.  Since the functions in Theorem~\ref{thm:local-sz-bd} are based on the measure of $J$ and not its width, we define the function
$$
h_0(y)=\left(d^{n+1}\lg w(I)-\frac{d^{n+1}}{n}\lg y +d^n\tau\right)k_0(d,\tau,n)
$$
where $k_0(d,\tau,n)$ is the maximum value over $I$ of the suppressed terms in the $\sOB$.  We observe that $h_0(\mu(J))$ is an upper bound on the bit-cost to compute the Taylor shift for the $C_0$ test for $J$.  

On the other hand, if $J$ passes $C_1$, since $g$ is a degree $2d-2$ polynomial in $n$ variables whose coefficients have maximum bit-size $O(\tau+d+\lg n)$, it follows that the charge associated to $J$ for the $C_1$ test is $\sOB(2^{2n}d^{2n+1}\lg w(I)-2^nd^{2n+1}\lg w(J)+2^{2n}d^{2n}(\tau+d+\lg n))$, which simplifies to $\sOB(2^{2n}d^{2n+1}\lg w(I)-2^{2n}d^{2n+1}\lg w(J)+2^{2n}d^{2n}\tau)$.  As above, we define the function
$$
h_1(y)=\left(2^{2n}d^{2n+1}\lg w(I)-\frac{2^{2n}d^{2n+1}}{n}\lg y+2^{2n}d^{2n}\tau\right)k_1(d,\tau,n)
$$
where $k_1(d,\tau,n)$ is the maximum value over $I$ of the suppressed terms in the $\sOB$.  We observe that $h_1(\mu(J))$ is an upper bound on the bit-cost to compute the Taylor shift for the $C_1$ test for $J$.

We use these two functions along with $G_0$ and $G_1$ as defined in Section~\ref{sec:adaptive-bounds} to develop adaptive bounds on the bit-complexity of the \PV\ algorithm as follows:

\begin{theorem}\label{thm:bit-CA}
Let $f\in\mathbb{Z}[x_1,\dots,x_n]$ be smooth and of degree $d$ and $I\subseteq\mathbb{R}^n$ be an axis-aligned $n$-dimensional cube whose corners have integral coordinates.  Let $\tau=\lg H$ be the maximum bit-size of the coefficients of $f$ and the corners of $I$.  The overall bit-complexity of the \PV\ algorithm is the maximum of $h_0(w(I)^n)$, $h_1(w(I)^n)$, and
$$
2^n\int_I\min\left\{\frac{h_0(2^{-n}G_0(x))}{G_0(x)},\frac{h_1(2^{-n}G_1(x))}{G_1(x)}\right\}dV_n.
$$
\end{theorem}

\section{Examples}
\label{sec:examples}  

The bounds in Theorems~\ref{thm:PV:steps} and~\ref{theorem:PV:CA:complexity} are both exponential with respect to the degree of the polynomial and the number of variables.  They remain exponential even if we assume that the number of variables is constant.  In~\cite{Plantinga:2004}, the authors show that for several examples the computation time is efficient in practice. The following examples illustrate that:
\begin{itemize}
\item The exponential behavior is optimal, up to constants in the exponents and
\item In particular cases, the complexity is provably better than the worst-case.
\end{itemize}

\begin{figure}[htb]
\centering
\includegraphics{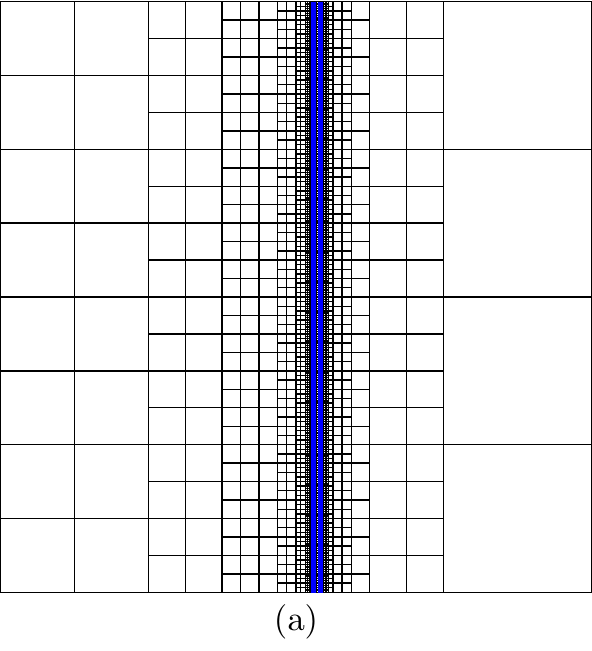}
\quad
\includegraphics{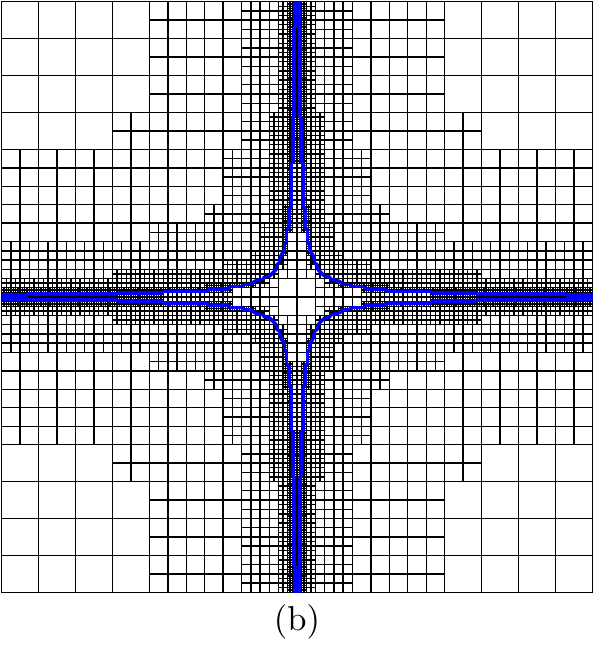}
\caption{(a) The output of the \PV\ algorithm for $f(x,y)=(x^d-2(ax-1)^2)(x^d-(ax-1)^2)$.  The solutions to $f(x,y)=0$ are close vertical lines (the illustrated case is when $d=3$ and $a=3$).  The width of boxes between vertical lines is at most $2a^{-\frac{d}{2}-1}$ and they extend the entire length of the initial region.  The number of regions is bounded from above by $\Omega(wa^{\frac{d}{2}+1})$ where $w=w(I)$ is the width of he initial region.  (b) The output of the \PV\ algorithm on $1000x^4y^4-1$.  We observe that the near-singularity at $(0,0)$ does not cause exponentially many subdivisions.  Instead, the pair of curves with the same asymptote contribute to this behavior since the width of boxes along the horizontal (vertical) axis must be less than the vertical (horizontal) distance between the two branches.
\label{fig:pv-curve-asymptotes}}
\end{figure}

\begin{lemma}\label{lem:PV-tight}
The bound of Theorem~\ref{thm:PV:steps} is asymptotically tight.
\end{lemma}
\begin{proof}
Following the construction in~\cite{Eigenwillig:Descartes}, consider the Mignotte polynomial $P(x)=x^d-2(ax-1)^2$ and the related polynomial $P_2(x)=x^d-(ax-1)^2$ where $a$ is a sufficiently large positive integer.  The product $P(x)P_2(x)$ is of degree $2d$ and the largest coefficient is of size $2a^4$.  In~\cite{Eigenwillig:Descartes}, it is shown that the product $P(x)P_2(x)$ has (at least) three roots in the interval $(a^{-1}-h,a^{-1}+h)$ where $h=a^{-d/2-1}$, see Figure~\ref{fig:pv-curve-asymptotes}(a).  Treating $P(x)P_2(x)$ as a polynomial in $n$ variables, we see that the \PV\ algorithm to approximate the variety in an $n$-dimensional cube $I$ of side length $w=w(I)$ requires subdividing until the regions have side length at most $2h$ to separate the three vertical hyperplanes in the interval $(a^{-1}-h,a^{-1}+h)$.  Since this occurs along an entire hyperplane of the input region, the number of small boxes is, at least, $\frac{w^{n-1}}{2h}=\frac{1}{2}w(I)^{n-1}a^{d/2+1}$, which is exponential in both the size of the input region and the size of the coefficients of the polynomial. 
\end{proof}

The previous example, while illustrating that the bounds are tight, raises the question of whether exponential behavior is due to the fact that the example is a one-dimensional problem lifted to higher dimensions.  We now show provide an example that shows that this exponential behavior can be observed for a curve involving both $x$ and $y$ in two dimensions.  In particular, in Lemma~\ref{lem:PV-tight}, the exponential behavior in two-dimensions was caused by two curves which were close together, but had a curve of critical points between them.  We can mimic that behavior for a curve in two-dimensions by considering a situation where two local components of the curve share an asymptote.

\begin{example}\label{ex:planar:exponential}
Fix $\varepsilon>0$ and consider $f(x_1,x_2)$ of one of the following forms:
\begin{itemize}
\item $f(x_1,x_2)=x_1^{a_1}x_2^{a_2}-\varepsilon^{a_1+a_2}$ where $a_1$ and $a_2$ are both positive integers and at least one is even or
\item $f(x_1,x_2)=x_1^{a_1}x_2^{a_2}+\varepsilon^{a_1+a_2}$ where $a_1$ and $a_2$ are both positive integers and exactly one is even.
\end{itemize}
In either of these cases, the \PV\ algorithm produces exponentially many regions in the size of the input box and the size of the coefficients of the polynomial, see Figure~\ref{fig:pv-curve-asymptotes}(b).

Since all of the cases are similar, we focus on the case where $f(x_1,x_2)=x_1^{a_1}x_2^{a_2}-\varepsilon^{a_1+a_2}$ and $a_2$ is even.  In this case, we show that the number of regions which intersect the positive $x$-axis is exponential in the size of the input.  Since $\nabla f$ is zero on the positive $x$-axis, any box which is terminal and intersects the positive $x$-axis must satisfy Condition $C_0$.  For any positive $x$,
$$
\left(x,\pm\left(\frac{\varepsilon^{a_1+a_2}}{x^{a_1}}\right)^{1/a_2}\right)
$$
are points on the variety $V_{\mathbb{R}}(f)$.  Therefore, any region which is terminal and contains $(x,0)$ must have width at most 
\begin{equation}\label{eq:lowerbound:CA}
2\left(\frac{\varepsilon^{a_1+a_2}}{x^{a_1}}\right)^{1/a_2}
\end{equation}
since, otherwise, the region would contain a point of $V_{\mathbb{R}}(f)$ and could not satisfy Condition $C_0$.

Let $J$ be a terminal region which intersects the positive $x$-axis and let $[s_1,s_2]$ be the intersection of $J$ with the positive $x$-axis.  Then, consider the integral
\begin{equation}\label{eq:lowerbound:Integral}
\frac{1}{2}\int_{s_1}^{s_2}\left(\frac{x^{a_1}}{\varepsilon^{a_1+a_2}}\right)^{1/a_2}dx\leq \frac{w(J)}{2}\left(\frac{s_2^{a_1}}{\varepsilon^{a_1+a_2}}\right)^{1/a_2},
\end{equation}
where the inequality follows since the integrand is increasing.
Since $(s_2,0)\in J$, by the bound in Expression~(\ref{eq:lowerbound:CA}), it follows that Expression~(\ref{eq:lowerbound:Integral}) is at most $1$.

Suppose that the intersection of the initial region $I$ with the positive $x$-axis is $[r_1,r_2]$.  Then, by the bound on Integral~(\ref{eq:lowerbound:Integral}) from above, it follows that 
$$
\frac{1}{2}\int_{r_1}^{r_2}\left(\frac{x^{a_1}}{\varepsilon^{a_1+a_2}}\right)^{1/a_2}dx=\frac{a_2}{2(a_1+a_2)}\left(\left(\frac{r_2}{\varepsilon}\right)^{\frac{a_1+a_2}{a_2}}-\left(\frac{r_1}{\varepsilon}\right)^{\frac{a_1+a_2}{a_2}}\right)
$$
is a lower bound on the number of regions formed by the \PV\ algorithm along the positive $x$-axis.  This region count is exponential in both the size of the input region and the size of the coefficients of the polynomial.
\end{example}

We remark that the example above is intrinsically hard for the algorithm and it can be adapted to higher dimensions and applies even under a change of coordinates.  We also note that the exponential behavior does not come from the near singularity at $(0,0)$, but from the curves sharing asymptotes.  For the centered form, see Section \ref{Section:localsizebound}, the complex portions of the curve also affect subdivisions, so, when using the centered form for the tests $C_0$ and $C_1$, the exponential behavior from the analysis above can be extended for all positive integers $a_1$ and $a_2$ such that $a_1+a_2>2$.

\begin{figure}[htb]
\centering
\includegraphics{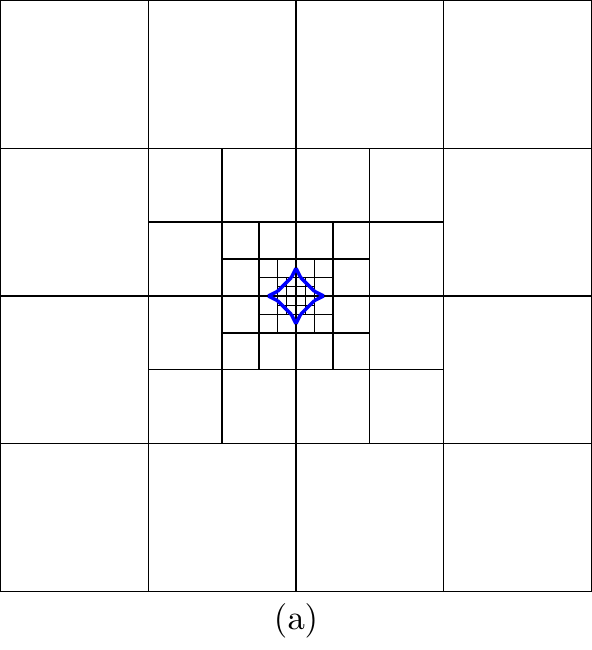}
\quad
\includegraphics{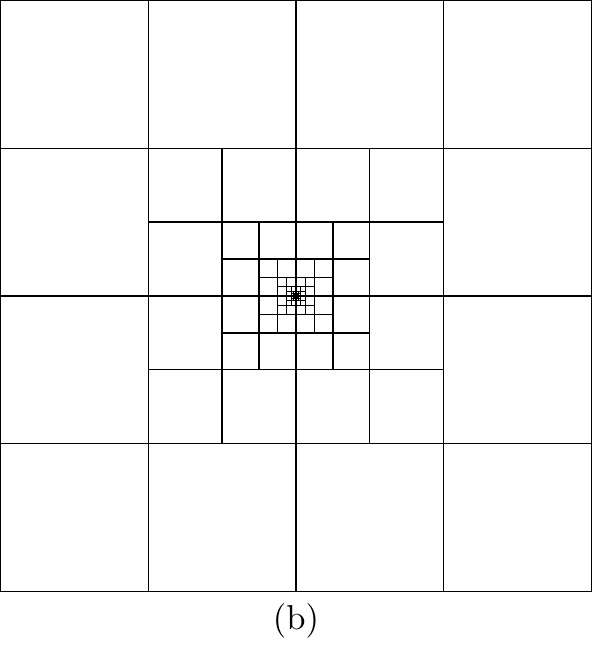}
\caption{(a) The output of the \PV\ algorithm for $f(x,y)=x^2+y^2-\varepsilon^2$.  The number of regions is bounded by $O(\lg(w)-\lg(\varepsilon))$.  (b) The output of the \PV\ algorithm for $f(x,y)=x^2+y^2+\varepsilon^2$.  The number of regions is bounded by $O(\lg(w)-\lg(\varepsilon))$.\label{fig:circle:b}}
\end{figure}

Even though our bounds are optimal, in practice, these are often quite pessimistic, as the actual separation bounds do not follow the worst case behavior.  We illustrate this better behavior in the following two examples:
\begin{example}
Fix $\varepsilon>0$ and consider  $f(x_1,x_2)=x_1^2+x_2^2+\varepsilon^2$.  Then, 
$$
\dist_{\CC}((x_1,x_2),f)=\sqrt{\frac{x_1^2+x_2^2}{2}+\varepsilon^2}
$$
and
$$
\dist_{\CC}((x_1,x_2,x_1,x_2),g)=\sqrt{x_1^2+x_2^2}.
$$  
Let $I$ be the initial input square where $w=w(I)$ is the width of $I$.  By substituting these bounds into Theorem~\ref{theorem:PV:CA:complexity}, we find that the number of regions constructed by the \PV\ algorithm is $O(\lg(w)-\lg(\varepsilon))$, see Figure~\ref{fig:circle:b}(a).
\end{example}

\begin{example}
Fix $\varepsilon>0$ and consider $f(x_1,x_2)=x_1^2+x_2^2-\varepsilon^2$.  Then, 
$$
\dist_{\CC}((x_1,x_2),f)=\begin{cases}\left|\sqrt{x_1^2+x_2^2}-\varepsilon\right|&x_1^2+x_2^2\leq 4\varepsilon^2\\\sqrt{\frac{x_1^2+x_2^2}{2}-\varepsilon^2}&x_1^2+x_2^2>4\varepsilon^2\end{cases}
$$ 
and 
$$\dist_{\CC}((x_1,x_2,x_1,x_2),g)=\sqrt{x_1^2+x_2^2}.$$  
Let $I$ be the initial input square where $w=w(I)$ is the width of $I$.  By substituting these bounds into Theorem~\ref{theorem:PV:CA:complexity}, we find that the number of regions constructed by the \PV\ algorithm is $O(\lg(w)-\lg(\varepsilon))$, see Figure~\ref{fig:circle:b}(b).

\end{example}

Moreover, we observe that for each of these examples, the minimum distance between $V_{\CC}(f^\Delta)$ and $V_{\CC}(g)$ is at most $\varepsilon$.  Therefore, a bound coming from Theorem~\ref{thm:PV:steps} would be much larger than the bound continuous amortization provides.

It remains an open question to deduce adaptive complexity bounds for the \PV\ algorithms from Theorem~\ref{theorem:PV:CA:complexity} based on geometric and {\em a priori} parameters.  We observe that since the complexity of the algorithm can be exponential in the inputs, the integral must be described in terms of more parameters than the degree of $f$ and the size of the coefficients of $f$.

\section*{Acknowledgments.}  The authors are grateful to Martin Sombra for his comments on the previous versions of this paper.  The authors are also grateful to Hoon Hong for his observations on how to generalize the initial results.  The authors thank Bernard Mourrain for his discussions about the difficult cases for the algorithm.

\bibliographystyle{plain}
\bibliography{PV_References}

\end{document}